\documentclass[lettersize,journal]{IEEEtran}
\usepackage{amsmath,amsfonts}
\usepackage{algorithmic}
\usepackage{algorithm}
\usepackage{array}
\usepackage[caption=false,font=normalsize,labelfont=sf,textfont=sf]{subfig}
\usepackage{textcomp}
\usepackage{stfloats}
\usepackage{url}
\usepackage{verbatim}
\usepackage{graphicx}
\usepackage{cite}
\def\BibTeX{{\rm B\kern-.05em{\sc i\kern-.025em b}\kern-.08em
    T\kern-.1667em\lower.7ex\hbox{E}\kern-.125emX}}
\usepackage{balance}
\IEEEoverridecommandlockouts
\usepackage{cite}
\usepackage{amsmath,amssymb,amsfonts}
\usepackage{algorithmic}
\usepackage{tikz}
\usepackage{graphicx}
\usepackage{textcomp}
\usepackage{xcolor}
\usepackage{lipsum} 

\usepackage{amsmath}

\usepackage{amsthm}
\theoremstyle{definition}
\newtheorem{theorem}{\normalfont\bfseries Theorem}
\newtheorem{lemma}{\normalfont\bfseries Lemma}
\newtheorem{example}{\normalfont\bfseries Example}
\newtheorem{definition}{\normalfont\bfseries Definition}

\newtheorem{assumption}{\normalfont\bfseries Assumption}

\usepackage{amsfonts}
\newcommand{\bx}{\mathbf{x}}
\newcommand{\by}{\mathbf{y}}
\newcommand{\bzero}{\mathbf{0}}
\renewcommand{\bf}{\mathbf{f}}
\newcommand{\bg}{\mathbf{g}}
\newcommand{\bu}{\mathbf{u}}
\newcommand{\bk}{\mathbf{k}}
\newcommand{\s}{\mathcal{S}}
\newcommand{\Sb}{\mathcal{S}_\mathrm{b}}
\newcommand{\SI}{\mathcal{S}_{\mathrm{I}}}
\newcommand{\Sns}{\mathcal{S}_{\mathrm{ns}}}
\newcommand{\R}{\mathbb{R}}
\newcommand{\Rn}{\mathbb{R}^n}
\newcommand{\Rm}{\mathbb{R}^m}
\newcommand{\Rp}{\mathbb{R}^p}
\newcommand{\hb}{h_\mathrm{b}}
\newcommand{\kd}{\mathbf{k}_{\mathrm{d}}}
\newcommand{\kb}{\mathbf{k}_\mathrm{b}}
\newcommand{\fb}{\mathbf{f}_\mathrm{b}}
\newcommand{\hkb}{\hat{\mathbf{k}}_\mathrm{b}}
\newcommand{\U}{\mathcal{U}}

\newcommand{\K}{\mathcal{K}}
\newcommand{\Ke}{\K^{\rm e}}
\newcommand{\boeta}{\boldsymbol{\eta}}
\newcommand{\bphi}{\boldsymbol{\Phi}}
\newcommand{\phib}{\boldsymbol{\varphi}_{\mathrm{b}}}
\newcommand{\bA}{\mathbf{A}}
\newcommand{\bP}{\mathbf{P}}
\newcommand{\bQ}{\mathbf{Q}}
\newcommand{\bI}{\mathbf{I}}
\newcommand{\kfl}{\bk_{\mathrm{FL}}}
\newcommand{\bK}{\mathbf{K}}

\begin{document}

\title{


 Braking within Barriers: Constructive Safety-Critical Control for Input-Constrained Vehicles via the Backup Set Method

%
\author{\IEEEauthorblockN{
Laszlo Gacsi,
Adam K. Kiss, and
Tamas G. Molnar
}
\thanks{This paper was supported by the János Bolyai Research Scholarship of the Hungarian Academy of Sciences and by the HUN-REN Hungarian Research Network.}%
\thanks{L. Gacsi and T. G. Molnar are with the Department of Mechanical Engineering, Wichita State University, Wichita, KS 67260, USA, {\tt\small lxgacsi@shockers.wichita.edu, tamas.molnar@wichita.edu}.}%
\thanks{Adam K. Kiss is with the HUN-REN–BME Dynamics of Machines Research Group, Department of Applied Mechanics, Budapest University of Technology and Economics, Budapest, Hungary, {\tt\small kiss\_a@mm.bme.hu}.}%
}
}



\maketitle

\begin{abstract}
This paper presents a safety-critical control framework to maintain bounded lateral motions for vehicles braking on asymmetric surfaces.
We synthesize a brake controller that assists drivers and guarantees safety against excessive lateral motions (i.e., prevents the vehicle from spinning out) while minimizing the stopping distance.
We address this safety-critical control problem in the presence of input constraints, since braking forces are limited by the available friction on the road.
We use backup control barrier functions for safe control design.
As this approach requires the construction of a backup set and a backup controller, we propose a novel, systematic method to creating valid backup set-backup controller pairs based on feedback linearization and continuous-time Lyapunov equations.
We use simple examples to demonstrate our proposed safety-critical control method.
Finally, we implement our approach on a four-wheel vehicle model for braking on asymmetric surfaces and present simulation results.
\end{abstract}

\def\abstractname{Note to Practitioners}
\begin{abstract}
This work is motivated by the challenge of braking vehicles on asymmetric surfaces where the left and right wheels experience different friction.
In such scenarios, ensuring that the vehicle does not spin out is as important as minimizing the stopping distance.
Existing techniques address this problem by allocating the braking forces among the wheels so as to stabilize the lateral motion, which may lead to increased stopping distance. In this paper, we propose a safety-critical control approach that constrains the lateral motions instead of stabilizing them, while accounting for the fact that braking forces are limited by the available friction.
Our approach provides the required braking forces that minimize the stopping distance while keeping the vehicle safe from spinning out.
For practitioners, our main contribution is a constructive, step-by-step method for designing safe controllers that rigorously handle physical actuator limits.
By combining established control theories in a novel way, our approach automates the design of these controllers, which previously relied on system-specific intuition.
A limitation of the method is that it requires an operating point where the controls are not saturated. For the asymmetric braking problem, this means that the driver's steering angle changes slowly.
\end{abstract}


\section{Introduction}

\IEEEPARstart{V}{ehicle} control systems have been designed to improve safety in a wide variety of driving conditions.
One of the most dangerous and intricate situations is braking on asymmetric, so-called split-$\mu$ surfaces, where the left and right wheels of road vehicles experience different coefficients of friction.
This scenario, illustrated in Fig. \ref{fig:musplitvehicle}, typically occurs for wet or frozen road sections, or when miscellaneous substances are present on the road.
The resulting difference in left and right braking forces generates an unexpected yaw moment which may find most drivers unprepared and may cause accidents if drivers react with inappropriate steering movements.

The challenge of safe split-$\mu$ braking arises from two contradictory control objectives: stopping over the shortest possible distance while maintaining lateral stability \cite{mirzaeinejad2014optimization}.
Minimizing the stopping distance requires maximizing the braking forces, up to the available road friction.
As the friction differs on the left and right sides, maximal braking forces would generate maximal yaw moment.
On the other hand, to eliminate the yaw moment and ensure lateral stability, both sides of the vehicle must brake equally.
This requires reducing the braking forces on one side, which increases the stopping distance.
Both of these approaches may endanger the vehicle and its driver.

\begin{figure}[t!]
\centerline{\includegraphics[width=0.458\textwidth]{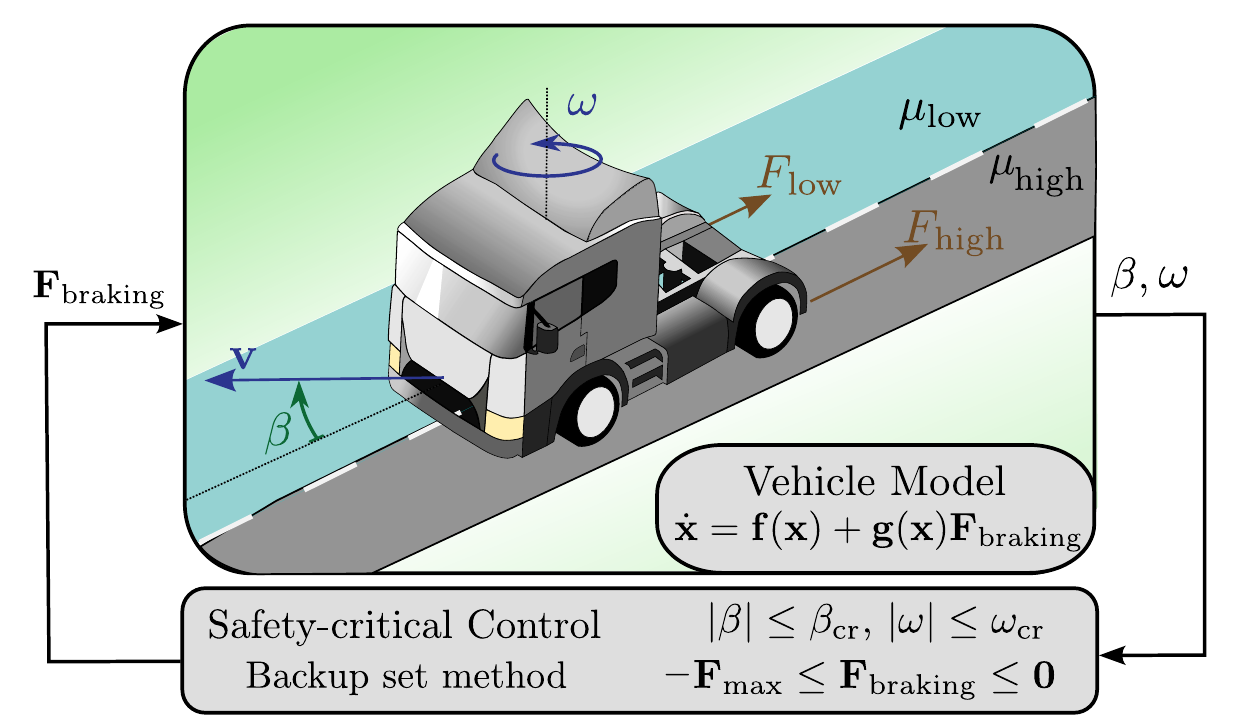}}
\vspace{-3mm}
\caption{A vehicle braking on an asymmetric (split-$\mu$) surface. The lateral kinematics (side slip angle and yaw rate) are kept within safe bounds with limited braking forces using the proposed safety-critical controller.}
\vspace{-3mm}
\label{fig:musplitvehicle}
\end{figure}

Vehicle control methods have the potential to mitigate the trade-off between the stopping distance and lateral stability.
To this end, brake controllers may exploit three fundamental quantities: brake pressure, brake pressure difference between left and right sides, and steering angle. Anti-lock braking systems (ABS) adjust the brake pressure in order to maximize the braking forces. Several studies \cite{yun2011brake,dong2017vehicle,kim2022control} draw upon this idea and adopt ABS to split-$\mu$ braking. 
Other methods \cite{zhou2011study,ahn2012modeling} use active steering control alongside ABS to ensure safe braking.
Alternatively, one may also regulate the brake pressure difference which is proportional to the yaw moment.
The work \cite{mirzaeinejad2014optimization} optimizes the yaw moment to reduce the yaw rate to a desired value, whereas \cite{feng2015integrated} determines a desired yaw moment for ABS
while using active steering control.
To allocate the yaw moment among the wheels, the former reference developed its own algorithm, while the latter applied a fuzzy controller.

In this paper, we take a novel approach to split-$\mu$ braking: we develop {\em safety-critical controllers}.
In contrast to previous works which aim to {\em stabilize} lateral motions via a specific yaw moment, we seek to {\em constrain} the vehicle's lateral kinematics.
Our goal is to prevent the yaw rate and side slip angle from rising to extreme values, instead of driving them to zero which would lead to equal braking on both sides and a drastic increase in the stopping distance.
We use safety-critical control techniques to keep
the yaw rate and side slip angle below safe bounds.
The proposed controllers provide the required braking forces, which may be applied to the wheels using an ABS.

In particular, we use the concept of control barrier functions (CBFs) \cite{ames2017cbf} to design safe controllers.
CBFs provide a state-of-the-art solution to a wide range of safety-critical control problems by ensuring that the state of the system is maintained inside a safe set in the state space.
This approach has demonstrated success in applications ranging from robotics
\cite{nguyen2021robust,molnar2022modelfree,garg2024advances}
to aerospace systems~\cite{hobbs2023rta}.
Vehicle control has also leveraged CBFs to ensure safety for both longitudinal and lateral motions, including adaptive cruise control~\cite{ames2014control, xu2017correctness, molnar2022safety} with experiments on connected automated vehicles~\cite{Alan2023AV}, obstacle avoidance scenarios~\cite{chen2018obstacle}, and~lane keeping~\cite{xu2017correctness, jiang2024safety}.

In split-$\mu$ braking, a key aspect of designing a safe controller is friction, which limits the available braking forces and thus leads to input constraints in addition to the lateral safety constraints.
The literature offers several CBF-based approaches to input-constrained safety-critical control.
These include input-constrained CBFs~\cite{agrawal2021safe} and sum-of-squares programming~\cite{wang2023safety} that reduce the size of safe sets to account for input bounds,
and integral CBFs~\cite{Ames2021icbf} that use forward integration to define safe sets for states and inputs separately.
A core challenge in these approaches is to prove the recursive feasibility of safe controllers.
To address this, the {\em backup set method} \cite{gurriet2018online} (or also called backup CBF method) predicts the future evolution of the system under a predefined backup controller and uses this information to calculate a safe and feasible control input.
This approach has been leveraged in a variety of applications, such as robotic systems~\cite{gurriet2020scalable},
including rovers~\cite{janwani2024learning}
and drones~\cite{singletary2022onboard},
spacecraft~\cite{dunlap2022comparing, vanwijk2024disturbancerobust},
and kinematic vehicle models~\cite{chen2021backup, rivera2024forward, rabiee2025soft}.


Thanks to its formal safety guarantees and feasibility under input constraints, the backup set method is an ideal candidate for tackling split-$\mu$ braking.
However, a prerequisite to this method is the synthesis of a backup controller and a corresponding backup set, which may be challenging for complex systems.
Most of the above-mentioned works used physical intuition to design backup sets and controllers, while \cite{wang2024warming} elaborated a formal construction for the special case of second-order systems.
These prior works, however, do not necessarily verify that the resulting backup set-backup controller pairs satisfy the properties required for safety and feasibility.
An algorithmic approach to the construction of {\em valid} backup set-backup controller pairs is missing from the literature.


In this paper, we develop a safety-critical control framework for vehicles braking on asymmetric surfaces.
We use the backup set method to keep the lateral kinematics (yaw rate and side slip angle) within safe bounds while satisfying input constraints arising from the physical limits of braking forces.
Our contributions are twofold.
First, we introduce a systematic method to generate backup set-backup controller pairs via feedback linearization and continuous-time Lyapunov equations.
This advances control theory for input-constrained safety-critical systems
by solving a key challenge for backup CBFs.
Second, we use this method to develop a safe controller for split-$\mu$ braking that prevents excessive lateral motions while bringing vehicles to a stop within a short distance.
We establish the controller for a four-wheel planar vehicle model with tire dynamics,
and we demonstrate the efficacy of our approach in simulation.
To the best of our knowledge, this is the first application of CBF theory to split-$\mu$ braking.


The rest of the paper is organized as follows. 
Section~\ref{sec:CBF} provides background on CBFs and the backup set method.
Section~\ref{sec:CTLE} presents our approach for constructing backup set-backup controller pairs, demonstrated on two simple examples.
Section~\ref{sec:splitmu} applies this method to synthesizing controllers for safe split-$\mu$ braking, which are implemented and simulated on a descriptive vehicle model.
Section \ref{sec:concl} closes with conclusions.



\section{Safety-critical Control with Input Constraint}
\label{sec:CBF}

Consider a control system in affine form:
\begin{equation}
    \dot{\bx}=\bf(\bx)+\bg(\bx)\bu\,,
    \label{eq:system}
\end{equation}
where ${\bx \in \Rn}$ is the state, 
${\bf : \Rn\to \Rn}$ and ${\bg : \Rn \to \R^{n \times m}}$ are smooth\footnote{A function is smooth if it is continuously differentiable as many times as necessary.
For a smooth function ${h : \Rn \to \R}$ and vector field ${\bf : \R^n \to \Rn}$, ${\nabla h : \Rn \to \Rn}$ denotes the gradient and ${\nabla\bf : \Rn \to \R^{n \times n}}$ is the Jacobian.} vector and matrix functions, respectively, ${\bu \in \U}$ is the input, and ${\U \subset \Rm}$ is the set of admissible inputs.
For simplicity, we assume that $\U$ is given by box constraints:
\begin{equation} \label{eq:inputset}
\begin{aligned}
    \U = \big\{ \bu\in \Rm : u_{\min,i} \leq u_i \leq u_{\max,i} \,, \
    \forall i \in \{ 1,2,\dots,m \}  \big\}\,,
\end{aligned}
\end{equation}
where $u_{\min,i}$ and $u_{\max,i}$ are the lower and upper bounds of the $i$-th input component $u_i$.
Consider a locally Lipschitz controller ${\bk : \Rn \to \U}$, ${\bu=\bk(\bx)}$, and the closed-loop dynamics:
\begin{equation}
    \dot{\bx}=\bf(\bx)+\bg(\bx)\bk(\bx)\,.
    \label{eq:closedloop}
\end{equation}
Then, with initial condition ${\bx(0) \in \Rn}$, the system \eqref{eq:closedloop} has a unique solution $\bx(t)$ that is assumed to exist for all ${t \geq 0}$. 

In order to maintain safety, the controller $\bk$ must be designed so that the solution of \eqref{eq:closedloop} remains inside a constraint set\footnote{We use the term {\em constraint set} for the user-defined set $\s$. This is not necessarily a {\em safe set} that can be made forward invariant with input constraints.}
${\s \subset \Rn}$. This is related to the forward invariance of $\s$. 

\begin{definition}
The dynamical system \eqref{eq:closedloop} is safe w.r.t.~$\s$ if the set $\s$ is forward invariant along \eqref{eq:closedloop}, meaning that
${\bx(0) \in \s \implies \bx(t) \in \s}$, ${\forall t \geq 0}$.
\end{definition}

\subsection{Background on control barrier functions}

Here, we overview the theory of control barrier functions for safe control design.
We define the set $\s$ as the zero-superlevel set of a smooth function ${h:\Rn\to\R}$:
\begin{align}
    \s&=\{ \bx\in \Rn : h(\bx)\geq0\}\,.
    \label{eq:safeset}
\end{align}
So far $h$ is merely a constraint function whose sign indicates whether a state is in the constraint set.
The following definition determines whether $h$ is a valid control barrier function (CBF).

\begin{definition}[\!\!\cite{ames2017cbf}]
Function $h$ is a {\em control barrier function} for \eqref{eq:system} on $\s$ if there exists\footnote{Function  ${\alpha : [0,a) \to \R}$, ${a>0}$, is of class-$\K$ (${\alpha \in \K}$) if it is continuous, strictly increasing, and ${\alpha(0)=0}$.
Function ${\alpha : (-b,a) \to \R}$, ${a,b>0}$ is of extended class-$\K$ (${\alpha \in \Ke}$) if it has the same properties.} ${\alpha \in \Ke}$ such that:
\begin{equation}
    \sup_{\bu \in \U}\,\dot{h}(\bx,\bu) > -\alpha\big(h(\bx)\big)
    \label{eq:forward}
\end{equation}
holds for all $\bx\in\s$, where the time derivative of $h$ is:
\begin{equation}
    \dot{h}(\bx,\bu) = \nabla h(\bx) \cdot \big( \mathbf{f}(\bx)+\mathbf{g}(\bx)\bu \big)\,.
\end{equation}
\end{definition} 

The existence of a CBF guarantees that the set $\s$ can be rendered forward invariant, as formally stated below.
\begin{theorem}[\!\!\cite{ames2017cbf}]
    \textit{If $h$ is a CBF for \eqref{eq:system} on $\s$, then any locally Lipschitz controller $\mathbf{u = k(x)}$ satisfying:
\begin{equation}
    \dot{h}(\mathbf{x,k(x)}) \geq -\alpha\big(h(\bx)\big)\,,
    \label{eq:forward_inv}
\end{equation}
for all ${\bx \in \s}$, renders~\eqref{eq:closedloop} safe w.r.t.~$\s$.}
\label{th1}
\end{theorem}

CBFs enable safety-critical controller synthesis. In essence, one needs to find a controller $\bk$ satisfying~\eqref{eq:forward_inv}.
This is possible, e.g., via optimization, by solving the quadratic program (QP):
\begin{equation}
\begin{aligned}
\mathbf{k}(\bx) = \, \underset{\bu\in\U}{\operatorname{argmin}} \quad & \|\bu-\kd(\bx)\|^2\,,\\
\textrm{s.t.} \quad & \dot{h}(\bx,\bu) \geq -\alpha(h(\bx))\,.
\label{eq:QP}
\end{aligned}
\end{equation}
In \eqref{eq:QP} we minimize the deviation of the input $\bu$ from a desired controller
${\kd : \Rn \to \U}$ such that $\s$ is made forward invariant. The CBF-QP in \eqref{eq:QP} is feasible if $h$ is a valid CBF as in \eqref{eq:forward}.

Finding a valid CBF, however, is a hard problem, especially with input constraints.
An arbitrary function $h$ may not be a CBF satisfying~\eqref{eq:forward}, and it may not be possible to render an arbitrary constraint set $\s$ forward invariant.
This motivated the introduction of the backup set method, which makes a subset of $\s$ forward invariant, even with input constraints.

\subsection{Backup set method}

The \textit{backup set method} \cite{gurriet2018online} has been proposed for addressing safety-critical control under input constraints. Its core idea is to predict the future state of the system under an input-constrained controller, called the backup controller, and ensure that the safe control task remains feasible with this controller.
A detailed proof of safety guarantees is given in \cite{gurriet2018online, molnar2023safety}.

The first step of this method is the construction of a valid {\em backup controller} and a corresponding {\em backup set}, which our work aims to facilitate.
 That is, one must define a subset of the constraint set $\s$, namely the backup set ${\Sb \subseteq \s}$:
\begin{equation}
    \Sb =\{ \bx\in \Rn : \hb(\bx)\geq0\}\,,
    \label{eq:backupset}
\end{equation}
with a smooth function ${\hb : \Rn \to \R}$. Simultaneously, one must also define an input-constrained, smooth backup controller ${\kb : \Rn \to \U}$
satisfying for all ${\bx \in \Rn}$ that:
\begin{equation} \label{eq:kb_def}
\begin{aligned}
    u_{\min,i} \leq\, & k_{\mathrm{b},i}(\bx)\leq u_{\max,i}, \quad
    \forall i\in\{ 1,2,\dots,m \} \,.
\end{aligned}
\end{equation}
Furthermore, the backup controller must render the backup set forward invariant along the corresponding closed-loop system:
\begin{equation}
    \dot{\bx}=\bf(\bx)+\bg(\bx)\kb(\bx)\triangleq\,\fb(\bx)\,,
    \label{eq:f_b}
\end{equation}
whose solution $\phib(t,\bx(0))$ is called the {\em backup flow}. 

The conditions for a valid backup set-backup controller pair are summarized as follows.
\begin{definition} \label{def:backup}
A set ${\Sb \subset \Rn}$ and a controller ${\kb : \Rn \to \Rm}$ are a {\em valid backup set -- backup controller pair} if:
\begin{itemize}
    \setlength{\itemindent}{12pt}
    \item[(C1)] ${\Sb \subseteq \s}$,
    \item[(C2)] ${\kb(\bx) \in \U}$ for all ${\bx \in \Rn}$,
    \item[(C3)] $\kb$ renders $\Sb$ forward invariant along~\eqref{eq:f_b}.
\end{itemize}
\end{definition}

The primary challenge in applying the backup set method
is the non-trivial task of finding a valid backup set-backup controller pair that satisfies these conditions. Existing examples \cite{chen2021backup, singletary2022onboard, molnar2023safety, wang2024warming, van2025disturbance} typically use physics-based intuition to construct problem-specific backup sets and backup controllers, and they do not necessarily verify the forward invariance of the backup set.
Below we will address this challenge and propose a method to generate valid backup sets and controllers.

\begin{figure}[t!]
\centerline{\includegraphics[width=0.46\textwidth]{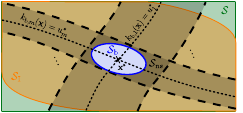}}
\vspace{-2mm}
\caption{
Illustration of the constraint set $\s$ (green), the backup set $\Sb$ (blue), and the invariant set $\SI$ (orange) obtained by enlarging the backup set.
The backup set, centered at point $\bx^*$ (black), is contained inside both the constraint set $\s$ and the set $\Sns$ (gray) where the backup controller does not saturate.
}
\label{fig:hb_position}
\vspace{-4mm}
\end{figure}

If one manages to find a backup set, it may be significantly smaller than the constraint set \cite{gurriet2018online, gurriet2020scalable}; see the illustration in Fig.~\ref{fig:hb_position}. This conservativeness can be resolved by enlarging the backup set to an invariant set ${\SI \subseteq \s}$ via forward prediction:
\begin{equation}
    \SI=
    \left\{ \bx \in \Rn :
    \begin{array}{l}
    \phib(\theta,\bx) \in \s, \ \forall \theta \in [0,T], \\
    \phib(T,\bx) \in \Sb
    \end{array}
    \right\}\,.
    \label{eq:SI}
\end{equation}
Accordingly, $\SI$ is the set of points $\bx$ from which the backup flow $\phib(\theta,\bx)$ evolves in the constraint set $\s$ for a duration ${\theta\in[0,T]}$ and arrives at the backup set $\Sb$ at time $T$.\
In this context, ${T \geq 0}$ is called the integration time or horizon.
With ${T=0}$, we obtain ${\SI=\Sb}$, and by increasing $T$, $\SI$ expands.

The following theorem establishes that, even under input constraints, the system~\eqref{eq:closedloop} can be rendered safe w.r.t.~the set $\SI$ in \eqref{eq:SI} which is a subset of the constraint set $\s$.

\begin{theorem}[\!\!\cite{molnar2023safety}]
\textit{Consider the system \eqref{eq:system}, the set $\s$ in \eqref{eq:safeset}, a backup set ${\Sb \subseteq \s}$ in \eqref{eq:backupset}, a backup controller ${\kb : \Rn \to \U}$ that renders $\Sb$ forward invariant along~\eqref{eq:f_b}, and the set $\SI$ in \eqref{eq:SI} with some ${T \geq 0}$.
Then, there exist ${\alpha,\alpha_\mathrm{b} \in \K}$ such that a controller ${\bk : \Rn \to \U}$ satisfying:
\begin{equation}
\begin{aligned}
    \dot{h}\big(\phib(\theta,\bx),\mathbf{k}(\bx)\big) &\geq -\alpha\big(h(\phib(\theta,\bx))\big)\,,
    \,\forall\theta\in[0,T]\,,\\
    \dot{h}_{\mathrm{b}} \big( \phib(T,\bx),\bk(\bx) \big) &\geq -\alpha_{\mathrm{b}} \big( \hb(\phib(T,\bx)) \big)\,,
\label{eq:backup_constr}
\end{aligned}
\end{equation}
for all ${\bx \in \SI}$ is guaranteed to exist. Moreover, any locally Lipschitz controller ${\bk :\Rn \to \U}$ that satisfies \eqref{eq:backup_constr} for all ${\bx \in \SI}$ renders ${\SI \subseteq \s}$ forward invariant along \eqref{eq:closedloop}.}
\label{th:2}
\end{theorem}

For controller synthesis, the CBF-QP \eqref{eq:QP} is extended with predictive constraints, resulting in the {\em backup CBF-QP}:
\begin{equation}
\begin{aligned}
\mathbf{k}(\bx) & = \underset{{\bu\in\U} }{\operatorname{argmin}} \quad  \|\bu-\mathbf{k}_{\mathrm{d}}(\bx)\|^2\,, \\
\textrm{s.t.} \quad & \dot{h} \big( \phib(\theta,\bx),\bu \big) \geq -\alpha\big(h(\phib(\theta,\bx))\big)\,, \
\forall \theta \!\in\! [0,T]\,, \\
&\dot{h}_{\mathrm{b}} \big( \phib(T,\bx),\bu \big) \geq -\alpha_{\mathrm{b}}\big(\hb(\phib(T,\bx))\big)\,,
\label{eq:QP2}
\end{aligned}
\end{equation}
which is feasible (with appropriate $\alpha$, $\alpha_{\mathrm{b}}$) and safe on the basis of Theorem \ref{th:2}.
The time derivatives in \eqref{eq:QP2} can be evaluated using the chain rule:
\begin{equation}
\begin{aligned}
    \dot{h} \big( \phib(\theta,\bx),\bu \big) & = \nabla h(\phib(\theta,\bx)) \!\cdot\! \bphi(\theta,\bx) \!\cdot\! \big(\bf(\bx) \!+\! \bg(\bx)\bu \big) \,,\\
    \dot{h}_{\mathrm{b}}(\phib(T,\bx),\bu) & = \nabla \hb(\phib(T,\bx)) \!\cdot\! \bphi(T,\bx) \!\cdot\! \big(\bf(\bx) \!+\! \bg(\bx)\bu \big) \,,\\
\end{aligned}
\end{equation}
where
${\bphi(\theta,\bx) \triangleq \frac{\partial \phib}{\partial \bx}(\theta,\bx)}$ is called the sensitivity matrix.
Typically, the flow $\phib(\theta,\bx)$ and its sensitivity $\bphi(\theta,\bx)$ are calculated by numerically solving the initial value problem:
\begin{align}
\begin{aligned}
    \frac{\partial \phib}{\partial \theta}(\theta,\bx) & = \fb\big(\phib(\theta,\bx)\big)\,,
    && \phib(0, \bx) = \bx\,, \\
    \frac{\partial \bphi}{\partial \theta}(\theta,\bx)  & = \nabla\fb\big(\phib(\theta,\bx)\big) \bphi(\theta,\bx)\,, && \ \bphi(0, \bx) = \bI\,, 
    \label{eq:IVP}
\end{aligned}
\end{align}
where $\bI$ is the ${n\times n}$ identity matrix.
For computational tractability, the continuous constraint in \eqref{eq:QP2} over the horizon $[0,\,T]$ is usually enforced at $N_\mathrm{c}$ discrete time-instances,
which may affect safety guarantees.
Rigorous methods for inter-sample safety~\cite{gurriet2020scalable, tan2025zero}, are beyond the scope of this paper.

\section{Synthesis of Backup Sets and Controllers}
\label{sec:CTLE}

Our main contribution is a systematic method for constructing valid backup set-backup controller pairs that rigorously handle input constraints. Our approach begins with defining the desired geometry of the backup set using the continuous-time Lyapunov equation \cite{Khalil2002}. This allows us to generate ellipsoid-shaped sets that are provably forward invariant, but for a simplified, "target" linear system. To bridge the gap to our original nonlinear system, we synthesize a backup controller via feedback linearization. The crucial insight for handling input constraints is that the linearization is only required inside the backup set. Thus, by choosing the backup set to be sufficiently small to lie within the region where the controller does not saturate, we can safely apply saturation outside the set. This helps to satisfy the input limits globally without invalidating the set's guaranteed forward invariance.


Our construction can be represented in the following form for feedback linearizable systems with invertible input matrix:
\begin{equation}
    \begin{aligned}
        \hb(\bx)&=c-V(\bx)\,,\\
        \kb(\bx)&=
        \mathrm{sat} \Big( \bg(\bx)^{-1} \big(-\bf(\bx) +\bA(\bx-\bx^*)  \big) \Big)\,,
    \end{aligned}
    \label{eq:main_idea}
\end{equation}
where ${c>0}$ is a constant, ${V: \Rn \to \R}$ is a Lyapunov function, 
$\bA\in\R^{n\times n}$ is a system matrix, $\bx^*$ is an equilibrium point, and $\mathrm{sat}(\cdot)$ saturates the backup controller to satisfy \eqref{eq:kb_def}.

In the following subsections, we revisit Lyapunov's method and feedback linearization, and we explain how to choose $c$, $V$, $\bA$, and $\bx^*$, as well as how to extend this idea to more general systems. Finally, we provide illustrative examples.

\subsection{Lyapunov equations for backup sets}
First we discuss the construction of backup sets for linear systems.
An important observation is that the shape of the backup set is irrelevant for the applicability of the method. This allows us to focus on finding a forward invariant backup set for a given backup controller---opposite to the QPs above that seek controllers ensuring forward invariance for given sets.

We use Lyapunov's method for the construction of forward invariant sets for linear systems in the form:
\begin{equation}
    \dot{\bx}=\bA(\bx-\bx^*)\,,
    \label{eq:lin}
\end{equation}
where ${\bx^* \in \Rn}$ is the equilibrium point and ${\bA \in \R^{n\times n}}$ is the system matrix.
Later we will design a backup controller that renders the closed-loop dynamics in the form of~\eqref{eq:lin}. By choosing matrix $\bA$ to be Hurwitz (so that its eigenvalues have negative real parts), the equilibrium point $\bx^*$ becomes stable for the linear system \eqref{eq:lin}.
The existence of a Lyapunov function is then guaranteed by a converse Lyapunov theorem.
\begin{lemma}[\!\!{\cite[Thm 4.17]{Khalil2002}}] \label{lemma:1}
\textit{
If the equilibrium point $\bx^*$ of the linear system \eqref{eq:lin} is stable, then there exists a smooth Lyapunov function $V:\Rn\to\R$ such that:
\begin{equation}
\begin{aligned}
    V(\bx^*)=0\,, \quad
    & V(\bx)>0\,, \
    \forall \bx \in \Rn \setminus \{\bx^*\}\,,\\
    & \dot{V}(\bx) \leq 0\,, \
    \forall \bx \in \Rn\,,
    \label{eq:lyapunov}
\end{aligned}
\end{equation}
with ${\dot{V}(\bx) = \nabla V(\bx) \cdot \bA(\bx - \bx^*)}$,
and the sublevel sets of $V$:
\begin{equation}
\Omega_c=\{\bx\in\Rn : V(\bx)\leq c \}\,,
\label{eq:omega_c}
\end{equation}
are forward invariant along~\eqref{eq:lin} for any ${c > 0}$.
}
\end{lemma}
Lemma \ref{lemma:1} states that invariant sets exist for the stable system \eqref{eq:lin}. To construct them, we use a quadratic Lyapunov function:
 \begin{equation} \label{eq:quadratic_Lyapunov}
    V(\bx)=(\bx-\bx^*)^\top \bP(\bx-\bx^*)\,
\end{equation}
with a positive definite matrix ${\bP \in \R^{n \times n}}$. This form of $V$ satisfies the positivity conditions in Lemma \ref{lemma:1}.
The time derivative of $V$ along the linear system \eqref{eq:lin} is:
\begin{equation}
    \dot{V}(\bx)=(\bx-\bx^*)^\top\big(\bA^\top \bP+\mathbf{PA}\big)(\bx-\bx^*)\,.
\end{equation}
If the matrix in the middle is negative semidefinite, then ${\dot{V}(\bx)\leq 0}$ holds for all ${\bx \in \R^n}$, and $\Omega_c$ in \eqref{eq:omega_c} is forward invariant per Lemma \ref{lemma:1}. This condition is enforced by satisfying the {\em continuous-time Lyapunov equation} (CTLE):
\begin{equation}
\bA^\top\bP+\mathbf{PA}=-\bQ\,,
    \label{eq:ctle}
\end{equation}
where $\bQ\in\R^{n\times n}$ is chosen to be a symmetric positive definite (SPD) matrix. 
The following lemma connects the stability of $\bA$ to the solvability of the CTLE.
\begin{lemma}[\!\!\cite{Khalil2002}] \label{lemma:2}
\textit{
If $\bA$ is Hurwitz and $\bQ$ is symmetric positive definite, the CTLE~\eqref{eq:ctle} has a unique symmetric positive definite solution for $\bP$. With this solution, $V$ in~\eqref{eq:quadratic_Lyapunov} satisfies \eqref{eq:lyapunov}, thus
$\Omega_c$ in \eqref{eq:omega_c} is forward invariant along \eqref{eq:lin}.
}
\end{lemma}

With Lemma \ref{lemma:2}, we can construct forward invariant backup sets for \eqref{eq:lin} by setting ${\Sb \triangleq \Omega_c}$.
Thus $\Sb$ is given by~\eqref{eq:backupset} with:
\begin{equation} \label{eq:hb_simple}
    \hb(\bx)=c-V(\bx)=c-(\bx-\bx^*)^\top \bP(\bx-\bx^*)\,,
\end{equation} 
where one must choose a symmetric positive definite $\bQ$ (such as ${\bQ=\bI}$) and solve the CTLE \eqref{eq:ctle} to obtain $\bP$.
This gives an ellipsoid-shaped set as shown with blue in Fig.~\ref{fig:hb_position}.
The center of this set is the point ${\bx^* \in \Sb}$.
The size of the set depends on $c$ whereas its shape and orientation are determined by $\bP$.

To obtain a {\em valid} backup set, as in Definition~\ref{def:backup}, $\Sb$ must be forward invariant.
Therefore, based on Lemma~\ref{lemma:2}, $\bA$ must be Hurwitz (while $\bx^*$ and $c$ do not affect invariance).
Additionally, the backup set $\Sb$ must be a subset of the constraint set $\s$ in~\eqref{eq:safeset}.
Thus, its center $\bx^*$ must be in the interior of $\s$:
\begin{equation} \label{eq:equilibrium_safety}
    h(\bx^*)>0,
\end{equation}
while its size $c$ must be small enough.
We provide the suitable choices of $\bA$, $\bx^*$, and $c$ in the context of the backup controller. 

\subsection{Feedback linearization for backup controllers}
\label{sec:feedback_lin}

Having established the form of the backup set, now the backup controller needs to be developed. Evidently, the objective is to transform the original nonlinear system in \eqref{eq:system} into the linear one in \eqref{eq:lin}. This is achieved via \textit{feedback linearization}, that stabilizes the system around the equilibrium point $\bx^*$.

We first present feedback linearization under the following strong assumption which will be relaxed in the next subsection.
\begin{assumption} \label{assum:feedback_lin}
The system~\eqref{eq:system} is fully actuated (${n=m}$) and the input matrix $\bg(\bx)$ is invertible for all ${\bx \in \Rn}$.
\end{assumption}

With this assumption, the system~\eqref{eq:system} is feedback linearizable, and the feedback linearization controller is:
\begin{equation} \label{eq:feedback_lin_simple}
    \kfl(\bx) = \bg(\bx)^{-1} \big( -\bf(\bx) + \bA (\bx - \bx^*) \big).
\end{equation}
Note that this controller may violate the input bounds, i.e., there may exist ${\bx \in \Rn}$ such that ${\kfl(\bx) \notin \U}$.
The states where controller~\eqref{eq:feedback_lin_simple} satisfies the input bounds are given by: 
\begin{equation} \label{eq:no_saturation_set}
    \Sns = \big\{ \bx\in \Rn : \kfl(\bx) \in \U \big\}\,.
\end{equation}

As outlined in \eqref{eq:main_idea}, we establish the backup controller by saturating the feedback linearization controller\footnote{To ensure the smoothness of the backup controller $\kb$, a smooth counterpart of the saturation function could be used. For simplicity of exposition, we present our results for the nonsmooth saturation function.}:
\begin{equation} \label{eq:backup_controller}
\begin{split}
    \kb(\bx) &= \mathrm{sat}\big(\kfl(\bx) \big)\,, \\
    k_{\mathrm{b},i}(\bx) &=
    \begin{cases}
        u_{\max,i} &\mathrm{if}\,\,k_{\mathrm{FL},i}(\bx)>u_{\max,i}\,,\\
        k_{\mathrm{FL},i} (\bx)&\mathrm{if}\,\,u_{\min,i}\leq k_{\mathrm{FL},i}(\bx)\leq u_{\max,i}\,,\\
        u_{\min,i} &\mathrm{if} \,\, k_{\mathrm{FL},i}(\bx)<u_{\min,i}\,,
    \end{cases}
\end{split}
\end{equation}
for all ${i \in \{1, 2, \dots, m\}}$, with the saturation function $\mathrm{sat}$.
Note that ${\kb(\bx) = \kfl(\bx) \iff \bx \in \Sns}$.

Despite the saturation, the backup controller $\kb$ can render the backup set $\Sb$ forward invariant if the saturation happens outside $\Sb$.
This is formally stated by the following lemma.

\begin{lemma} \label{lem:Sb_forward_inv}
\textit{
Consider the system~\eqref{eq:system} with Assumption~\ref{assum:feedback_lin}, the backup set $\Sb$ given by~\eqref{eq:backupset} and~\eqref{eq:hb_simple}, the backup controller $\kb$ given by~\eqref{eq:feedback_lin_simple} and~\eqref{eq:backup_controller}, and the set $\Sns$ in~\eqref{eq:no_saturation_set}.
Let $\bA$ be Hurwitz and $\bP$ satisfy the CTLE~\eqref{eq:ctle}.
If ${\Sb \subseteq \Sns}$, then $\Sb$ is forward invariant along~\eqref{eq:f_b}.
}
\end{lemma}
\begin{proof}
Because ${\Sb \subseteq \Sns}$, we have ${\kb(\bx) = \kfl(\bx)}$ for all ${\bx \in \Sb}$.
Then, by substituting~\eqref{eq:feedback_lin_simple}, the backup system~\eqref{eq:f_b} simplifies to the linear system~\eqref{eq:lin} for all ${\bx \in \Sb}$.
Therefore, the backup set is forward invariant  based on Lemma~\ref{lemma:2}.
\end{proof}

According to Lemma~\ref{lem:Sb_forward_inv}, the backup set $\Sb$ can be rendered forward invariant if it is a subset of the set $\Sns$ associated with no saturation for the backup controller.
This is illustrated in Fig.~\ref{fig:hb_position}. 
In order to satisfy this condition, i.e., to ensure ${\Sb \subseteq \Sns}$, the center $\bx^*$ of $\Sb$ must be in the interior of $\Sns$:
\begin{equation} \label{eq:equilibrium_no_saturation}
    u_{\min,i}<k_{\mathrm{FL},i}(\bx^*)<u_{\max,i}\,, \quad
    \forall i \in \{1, 2, \dots, m\}\,;
\end{equation}
cf.~\eqref{eq:inputset} and~\eqref{eq:no_saturation_set}.
In other words, the backup controller must not saturate at the equilibrium point $\bx^*$.
We introduce the notation ${\bu^* \triangleq \kfl(\bx^*)}$ for the corresponding input, which satisfies:
\begin{equation}
    \bf(\bx^*)+\bg(\bx^*)\bu^*=\bzero\,.
    \label{eq:equilibrium}
\end{equation}


Overall, the equilibrium $\bx^*$ must be chosen so that both~\eqref{eq:equilibrium_safety} and~\eqref{eq:equilibrium_no_saturation} hold, that is, $\bx^*$ must lie inside the intersection of $\s$ and $\Sns$ as shown in Fig.~\ref{fig:hb_position}.
This necessitates that such a point $\bx^*$ exists in the first place, i.e., $\s$ and $\Sns$ must have an intersection.
This is stated formally by the following assumption.
Note that whether this assumption holds depends on $\bf$, $\bg$, and $h$ defining the safety-critical control problem, because~\eqref{eq:equilibrium_safety} contains the constraint function $h$, whereas~\eqref{eq:equilibrium_no_saturation} involves the dynamics $\bf$ and $\bg$ through ${\kfl(\bx^*) = -\bg(\bx^*)^{-1}\bf(\bx^*)}$.

\begin{assumption} \label{assum:nonempty}
For the system~\eqref{eq:system} and the safety constraint~\eqref{eq:safeset}, there exists ${\bx^* \in \R^n}$ satisfying~\eqref{eq:equilibrium_safety} and~\eqref{eq:equilibrium_no_saturation}.
\end{assumption}

The next theorem states our main result that the proposed construction leads to a valid backup set-backup controller pair, and summarizes the criteria for choosing $\bA$, $\bP$, $\bx^*$, and $c$.

\begin{theorem}
\textit{
Consider the system~\eqref{eq:system} with Assumption~\ref{assum:feedback_lin}, the constraint set $\s$ in~\eqref{eq:safeset} with Assumption~\ref{assum:nonempty}, the backup set $\Sb$ given by~\eqref{eq:backupset} and~\eqref{eq:hb_simple}, the backup controller $\kb$ given by~\eqref{eq:feedback_lin_simple} and~\eqref{eq:backup_controller}, and the set $\Sns$ in~\eqref{eq:no_saturation_set}.
Let ${\bA \in \R^{n \times n}}$, ${\bP \in \R^{n \times n}}$, ${\bx^* \in \R^n}$, and ${c>0}$ be chosen such that:
\begin{itemize}
    \item $\bA$ is Hurwitz,
    \item $\bP$ is the solution of~\eqref{eq:ctle} with a SPD matrix ${\bQ \in \R^{n \times n}}$,
    \item $\bx^*$ satisfies~\eqref{eq:equilibrium_safety} and~\eqref{eq:equilibrium_no_saturation},
    \item $c$ is small enough so that ${\Sb \subseteq \s}$ and ${\Sb \subseteq \Sns}$.
\end{itemize}
Then $\Sb$ and $\kb$ are a valid backup set-backup controller pair, satisfying conditions (C1)-(C3) in Definition~\ref{def:backup}.
}
\end{theorem}

\begin{proof}
Condition (C1) holds as $c$ is chosen such that ${\Sb \subseteq \s}$.
Condition (C2) is satisfied by the construction of $\kb$, due to the saturation in~\eqref{eq:backup_controller}.
Condition (C3) holds because of Lemma~\ref{lem:Sb_forward_inv}, provided that $\bA$ is Hurwitz, $\bP$ satisfies~\eqref{eq:ctle}, and $c$ is chosen so that ${\Sb \subseteq \Sns}$.
Since $\bx^*$ satisfies~\eqref{eq:equilibrium_safety} and~\eqref{eq:equilibrium_no_saturation}, it ensures that there exists ${c>0}$ such that ${\Sb \subseteq \s}$ and ${\Sb \subseteq \Sns}$.
\end{proof}

\subsection{Design considerations and practical guidelines}

The overall procedure of input-constrained safety-critical control is summarized in Fig.~\ref{fig:block_method}.
First, we design the proposed backup set, as in~\eqref{eq:backupset} and~\eqref{eq:hb_simple}, and backup controller, as in~\eqref{eq:feedback_lin_simple} and~\eqref{eq:backup_controller}, through the following steps.
\begin{enumerate}
    \item Choose an equilibrium $\bx^*$ satisfying~\eqref{eq:equilibrium_safety} and~\eqref{eq:equilibrium_no_saturation}.
    \item Choose a Hurwitz $\bA$.
    \item Choose a SPD $\bQ$ and solve the CTLE~\eqref{eq:ctle} for $\bP$.
    \item Choose $c$ so that ${\Sb \subseteq \s}$ and ${\Sb \subseteq \Sns}$.
\end{enumerate}
These steps are performed only once (offline).
Then, at each time moment during safety-critical control (online), we proceed with the steps of the backup set method as in~\cite{gurriet2018online}.
\begin{enumerate}
    \setcounter{enumi}{4}
    \item Solve the initial value problem~\eqref{eq:IVP}.
    \item Solve the backup CBF-QP in~\eqref{eq:QP2}.
\end{enumerate}

\begin{figure}[t!]
\centerline{\includegraphics[width=0.5\textwidth]{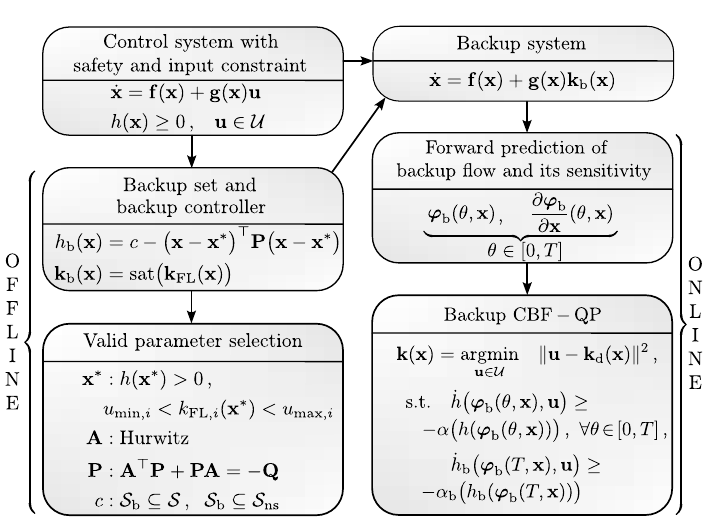}}
\vspace{-2mm}
\caption{Flow chart of the proposed method, summarizing the steps to generate valid backup set-backup controller pairs for use in the backup set method, which ultimately enables input-constrained safety-critical control.}
\vspace{-4mm}
\label{fig:block_method}
\end{figure}

While these steps provide a constructive framework for synthesizing a valid $(\Sb, \kb)$ pair, they also offer considerable design freedom. The choice of the parameters $(\bx^*, \bA, \bQ, c)$ collectively determine the final invariant set $\SI$.
The parameter selection is guided by the following principles and trade-offs.

Typically, the practical goal is to select these parameters so that they yield the largest possible set $\SI$, as this translates into a less conservative controller.
The size and shape of $\SI$ is determined by the interplay of the backup controller $\kb$, backup set $\Sb$, and prediction horizon $T$. Since a long horizon is computationally expensive and sensitive to model errors, it is practically advantageous to construct the largest possible $\Sb$ and obtain a large $\SI$ with a short and robust prediction.

To enlarge $\Sb$, the equilibrium point $\bx^*$ should be chosen strategically to center the backup set. A choice deep within the intersection of the constraint set $\s$ and the no-saturation region $\Sns$ maximizes the available space for $\Sb$ to be expanded.
A minimal-effort steady state is often a good candidate.

The system matrix $\bA$ dictates the desired closed-loop dynamics and presents a key trade-off. While an ''aggressive'' $\bA$ (with fast poles) promises rapid stabilization, the total control effort of $\kfl$ in~\eqref{eq:feedback_lin_simple}
may become large. This shrinks the no-saturation region $\Sns$, thus limiting the maximum possible size of $\Sb$. The optimal choice of $\bA$ is often system-dependent, balancing the desired convergence rate with having low total control effort over a large region.
The matrix $\bQ$ then shapes the geometry of the resulting ellipsoid $\Sb$ via the CTLE in~\eqref{eq:ctle}.

Finally, the choice of the scalar $c$ can maximize the size of the backup set $\Sb$ while satisfying the conditions ${\Sb \subseteq \s}$ and ${\Sb \subseteq \Sns}$. The maximal $c$ is found where the boundary of the ellipsoid $\Sb$ becomes tangent to the boundary of either $\s$ or $\Sns$. This geometric condition can be formulated as a constrained optimization problem and solved using the method of Lagrange multipliers~\cite{bertsekas2014constrained}. The core idea, that the gradients of the respective surfaces are parallel at the point of tangency, leads to a system of nonlinear algebraic equations. A similar principle of finding the largest invariant set via tangency conditions was applied to systems with time delays in~\cite{Kiss2021}.

To illustrate the full control synthesis procedure, we present a simple example: a fully-actuated scalar system.

\begin{example} \label{ex:scalar}

Let us consider the following scalar system:
\begin{equation}
    \dot{x}=x^3 + u\,,
    \label{eq:ex1}
\end{equation}
with ${x \in \R}$, ${u \in \R}$, ${f(x)=x^3}$, and ${g(x)=1}$.
We intend to keep the state in the interval ${-1\leq x\leq 1}$.
This leads to the constraint set $\s$ in \eqref{eq:safeset} with ${h(x)=1-x^2}$.
We consider the input bounds ${u_{\min} \leq u \leq u_{\max}}$ with ${u_{\min}=-0.5}$ and ${u_{\max}=0.75}$.
We seek to design a safe controller $k$ by modifying the desired controller ${k_{\mathrm{d}}(x)=0}$.
Note that $k_{\mathrm{d}}$ would let the state go unbounded and leave the constraint set if ${x(0) \neq 0}$.

The system~\eqref{eq:ex1} satisfies Assumption~\ref{assum:feedback_lin}, thus we proceed to design a backup set and a backup controller with the proposed method.
The backup set is given by:
\begin{equation}
    \hb(x)=c-P\left(x-x^*\right)^2\,,
    \label{eq:ex1_hb}
\end{equation}
with ${c>0}$ and ${P>0}$.
The backup controller, that linearizes the system around the equilibrium point $x^*$, becomes:
\begin{equation}
    k_\mathrm{b}(x)=\mathrm{sat}\big(-x^3-K(x-x^*) \big)\,,
    \label{eq:ex1_kb}
\end{equation}
with a gain ${K>0}$.
The backup controller is saturated due to the input constraint:
\begin{equation}
    \mathrm{sat}(u)=\begin{cases}
    u_{\max} &\mathrm{if}\,\,u>u_{\max}\,,\\
    u &\mathrm{if}\,\,u_{\min}\leq u\leq u_{\max}\,,\\
    u_{\min} &\mathrm{if} \,\, u<u_{\min}\,.
    \end{cases}
\end{equation}

\begin{figure}
\centering
\includegraphics[width=0.91\linewidth]{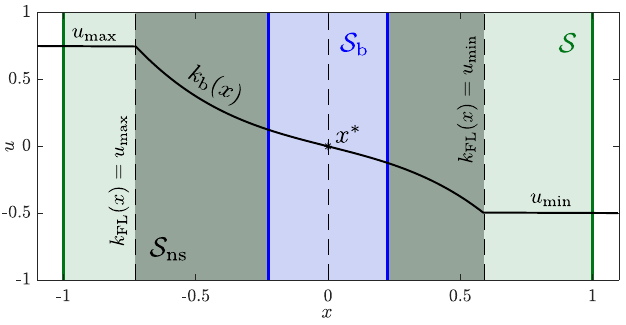}
\vspace{-2mm}
\caption{Illustration of the constraint set $\s$ (green), backup set $\Sb$ (blue), backup controller $k_{\mathrm{b}}$ (black), and input constraints $\Sns$ (gray) for Example~\ref{ex:scalar}.}
\vspace{-4mm}
\label{fig:setup}
\end{figure}

Now the parameters $c$, $P$, $x^*$, and $K$ need to be chosen.
We choose the equilibrium $x^*$ to be inside the constraint set with the corresponding input $u^*$ within the saturation limits:
\begin{equation}
    -1 < x^* < 1\,, \quad
    u_{\min} < u^* = -(x^*)^3 < u_{\max}\,.
\label{eq:intersection}
\end{equation}
This can be satisfied only if the two domains in~\eqref{eq:intersection} have an intersection, as stated in Assumption~\ref{assum:nonempty}.
Since ${u_{\min}<1}$ and ${u_{\max}>-1}$, this assumption holds.
By identifying ${A=-K}$, we may choose any positive gain ${K>0}$ to obtain a Hurwitz ${A<0}$.
Then, we solve the CTLE, ${AP \!+\! PA = -2KP = -Q}$, with ${Q=1}$ to obtain ${P=1/(2K)}$.
As for parameter $c$, it must be tuned so that the backup set is inside both the constraint set and the region where the backup controller does not saturate. 

One possible choice of backup set and backup controller is shown in Fig.~\ref{fig:setup}.
The equilibrium point is ${x^*=0}$, which yields ${u^*=0}$ and satisfies the conditions in \eqref{eq:intersection}.
The remaining parameters are set to ${K=0.5}$ and ${c=0.05}$, thus ${A=-0.5}$ and ${P=1}$.
Note that the backup set could be further expanded by increasing parameter $c$ up until reaching one of the boundaries of $\Sns$, given by $k_\mathrm{FL}(x)=u_{\min}$ and $k_\mathrm{FL}(x)=u_{\max}$ (black dashed lines). Beyond these bounds, the backup controller is saturated according to \eqref{eq:ex1_kb} as can be seen in Fig.~\ref{fig:setup}.
    
\begin{figure}
\centering
\includegraphics[width=0.99\linewidth]{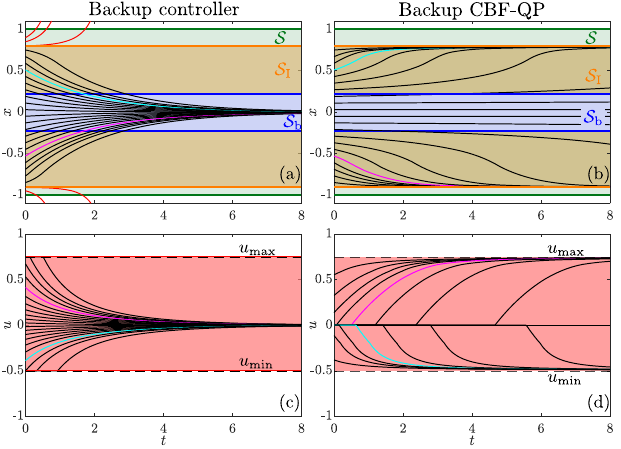}
\vspace{-7mm}
\caption{Trajectories of the system \eqref{eq:ex1} under the backup controller (a), and the backup CBF-QP controller (b) with the corresponding control input (c) and (d). The backup controller, given by \eqref{eq:ex1_kb}, maintains the forward invariance of the backup set $\Sb$ (blue) from \eqref{eq:ex1_hb}. The backup CBF-QP, given by \eqref{eq:QP2}, ensures the forward invariance of the expanded set $\SI$ (orange) that is within the constraint set $\s$ (green). Both controllers satisfy the input constraints.}
\vspace{-4mm}
\label{fig:backup_tr}
\end{figure}

Figure \ref{fig:backup_tr} shows numerical simulations of the system~\eqref{eq:ex1}.
In the left panels, the backup controller~\eqref{eq:ex1_kb} is directly executed.
Panel (a) depicts trajectories with various initial conditions within the constraint set $\s$. 
Trajectories launched from the backup set $\Sb$ remain in this set for all time, indicating forward invariance.
Furthermore, a larger set $\SI$ is also forward invariant, as indicated by the trajectories that start from and continue to evolve in this set (black trajectories).
The remaining trajectories (red trajectories) starting from ${\s\setminus\SI}$ leave the constraint set.
Panel (c) shows the corresponding inputs which satisfy the input constraints.
The magenta and cyan curves highlight two selected trajectories to show the correspondence between state and input.
Overall, the simulation results confirm that the backup set-backup controller pair is valid.

The right panels of Fig.~\ref{fig:backup_tr} show the behavior of the system~\eqref{eq:ex1} with the backup CBF-QP controller in~\eqref{eq:QP2}, which modifies the desired control input to a safe input while taking the backup flow into account.
This controller is designed with horizon $T=4$ discretized into ${N_\mathrm{c}=40}$ moments, and with class-$\K$ functions ${\alpha(h) = 0.5 h}$ and ${\alpha_{\mathrm{b}}(h_{\rm b}) = 0.25 h_{\rm b}}$.
Panel (b) highlights that the controller successfully keeps trajectories inside the invariant set $\SI$.
Moreover, while the backup controller drives the trajectories into the backup set in panel (a), the safety-critical controller in panel (b) is less conservative and allows trajectories to approach the boundary of the set $\SI$.
This highlights an important implication of the backup set method: the backup controller may not be used directly, but is rather a tool to implicitly define the invariant set $\SI$, which is then incorporated into the optimization problem in~\eqref{eq:QP2} to ensure feasibility.
Panel (d) shows that the proposed controller satisfies the input bounds, ultimately resulting in a safe and feasible control design.
\end{example}

\subsection{Generalization to output coordinates}
\label{sec:output_formulation}

The backup set-backup controller pair above is established for a certain class of feedback linearizable systems satisfying Assumption~\ref{assum:feedback_lin}, where the number of states and inputs is the same (${n=m}$) and the input matrix ${\bg(\bx)}$ is invertible.
Since this assumption is restrictive, now we relax it and extend our method to a wider class of systems.

In particular, we perform feedback linearization for a set of output coordinates $\by(\bx)$ rather than for the full state $\bx$.
We construct backup controllers that make the output dynamics linear and we create forward invariant backup sets based on the output.
For this construction, the notion of relative degree needs to be introduced for output coordinates. 

\begin{definition}{(\!\!\cite{isidori1985nonlinear})} The smooth function $\by : \Rn\to \Rp$ is said to have \textit{relative degree} $r\in\mathbb{N}$ with respect to \eqref{eq:system} if:
\begin{equation}
    \begin{aligned}
        (\mathrm{i})& \hspace{1cm} L_\bg L_\bf^i \by(\bx) =\bzero\,, \quad \forall i\in \{0,\dots,r-2 \}\,,\\
        (\mathrm{ii})& \hspace{1cm} \mathrm{rank}\big(L_\bg L_\bf^{r-1} \by(\bx)\big) =p\,,
    \end{aligned}
    \label{eq:reldeg}
\end{equation}
hold\footnote{The Lie derivative of $\by$ along $\bf$ is $L_\bf \by(\bx)\triangleq \nabla \by(\bx) \cdot \bf (\bx)$, while higher-order Lie derivatives are defined recursively: $L_\bf^i \by(\bx)\triangleq \nabla \big(L_\bf^{i-1} \by(\bx)\big)\cdot \bf(\bx)$.} for all ${\bx\in\Rn}$.
\end{definition}

With this definition, we replace Assumption~\ref{assum:feedback_lin} with a milder assumption that requires an output with a valid relative degree.

\begin{assumption} \label{assum:output_reldeg}
For the system~\eqref{eq:system}, there exists a smooth output ${\by: \Rn \to \Rp}$ with relative degree $r$, so that the number of outputs and inputs is the same (${p=m}$).
\end{assumption}

This means that $L_\bg L_\bf^{r-1} \by(\bx)$ is an invertible square matrix, as opposed to Assumption~\ref{assum:feedback_lin} where $\bg(\bx)$ is invertible.
Later we demonstrate in an example that this is a much milder condition.

\begin{figure*}
\centering
\includegraphics[width=0.925\textwidth]{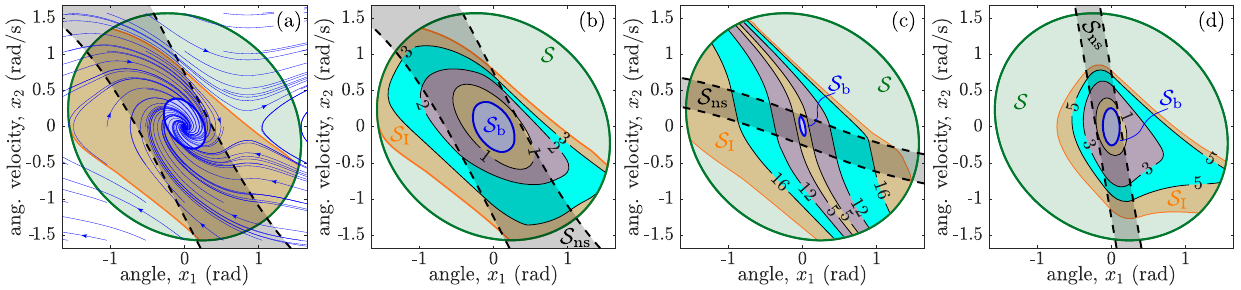}
\vspace{-2mm}
\caption{
Illustration of the constraint set $\s$ (green), backup set $\Sb$ (blue), input constraints $\Sns$ (gray), largest invariant set $\SI$ (orange), and invariant sets for various integration times (remaining filled contours) for Example~\ref{ex:pendulum}.
The effect of the backup controller gains on the shape of these sets is demonstrated for three different gain combinations: (a)-(b) ${K_1=1}$, ${K_2=1}$, ${c=0.1}$,
(c) ${K_1=1}$, ${K_2=5}$, ${c=0.0025}$, and
(d) ${K_1=5}$, ${K_2=1}$, ${c=0.04}$.}
\vspace{-4mm}
\label{fig:ex2_sets}
\end{figure*}

When the output $\by$ has relative degree $r$,
we take its derivatives and define a vector ${\boeta \in \R^{rp}}$:
\begin{equation}
    \boeta = \begin{bmatrix}
        \boeta_1\\
        \boeta_2\\
        \vdots\\
        \boeta_r
    \end{bmatrix}\triangleq \begin{bmatrix}
        \by(\bx) - \by(\bx^*)\\
        L_\bf \by(\bx)\\
        \vdots\\
        L_\bf^{r-1} \by(\bx)
    \end{bmatrix} \,,
\end{equation}
where ${\by(\bx^*)}$ is an equilibrium output with parameter $\bx^*$ and ${\boeta_{i+1} = L_\bf^i \by(\bx)}$ denotes the $i$-th time derivative of the output, ${i \in \{1, 2, \dots, r-1\}}$.
This results in the output dynamics:
\begin{equation}
    \dot{\boeta}=\begin{bmatrix}
        \dot{\boeta}_1\\
        \dot{\boeta}_2\\
        \vdots\\
        \dot{\boeta}_r
    \end{bmatrix}=\begin{bmatrix}
        \boeta_2\\
        \boeta_3\\
        \vdots\\
        L_\bf^{r} \by(\bx)+L_\bg L_\bf^{r-1} \by(\bx)\bu
    \end{bmatrix}\,.
    \label{eq:normal_form}
\end{equation}

The feedback linearization controller for~\eqref{eq:normal_form} becomes \cite{isidori1985nonlinear}:
\begin{equation}
    \kfl(\bx)=\big(L_\bg L_\bf^{r-1} \by(\bx)\big)^{-1}\bigg(-L_\bf^{r} \by(\bx) - \sum_{i=1}^{r}\bK_i\boeta_i\bigg)\,,
    \label{eq:feedback_lin}
\end{equation}
where $\bK_i$ are control gains.
Note that the decoupling matrix $L_\bg L_\bf^{r-1} \by(\bx)$ is invertible based on Assumption~\ref{assum:output_reldeg}.
With this controller, the closed-loop output dynamics are:
\begin{equation} \label{eq:output_closedloop}
    \dot{\boeta}=\underbrace{\begin{bmatrix}
        \bzero & \bI & \bzero & \dots & \bzero\\
        \bzero & \bzero & \bI & \dots & \bzero\\
        \vdots & \vdots & \vdots & \ddots & \vdots\\
        \bzero & \bzero & \bzero & \dots & \bI\\
        -\bK_1 & -\bK_2 & -\bK_3 & \dots & -\bK_r
    \end{bmatrix}}_{\bA}\boeta\,,
\end{equation}
where the matrix $\bA$ can be made Hurwitz by designing the gains $\bK_i$, for example, using pole placement.


Then, we propose to use the backup controller in~\eqref{eq:backup_controller} by saturating the feedback linearization controller in~\eqref{eq:feedback_lin}.
The corresponding backup set becomes:
\begin{equation}
    \hb(\bx)=c-\boeta^\top\bP\boeta\,,
    \label{eq:hb_finalform}
\end{equation}
cf.~\eqref{eq:hb_simple}.
This construction recovers the backup controller and backup set in Section~\ref{sec:feedback_lin} as a special case where the output is ${\by(\bx) = \bx}$ and it has relative degree one (i.e., ${\boeta = \bx-\bx^*}$).

To make this backup set-backup controller pair valid, one shall proceed with the following steps:
\begin{enumerate}
     \item Find an $m$-dimensional output $\by$ with a relative degree.
    \item Choose an equilibrium $\bx^*$ satisfying~\eqref{eq:equilibrium_safety} and~\eqref{eq:equilibrium_no_saturation}.
    \item Design the gains $\bK_i$ so that $\bA$ in~\eqref{eq:output_closedloop} is Hurwitz.
    \item Choose a SPD $\bQ$ and solve the CTLE~\eqref{eq:ctle} for $\bP$.
    \item Choose $c$ so that ${\Sb \subseteq \s}$ and ${\Sb \subseteq \Sns}$.
\end{enumerate}
%
%
In the following example we demonstrate this procedure with emphasis on the choice of the output coordinate.


% ...

\begin{example}
\label{ex:pendulum}
Consider the model of an inverted pendulum:
\begin{align}
    \begin{bmatrix}
        \dot{x}_1\\
        \dot{x}_2
    \end{bmatrix}=\begin{bmatrix}
        x_2\\
        \sin(x_1)
    \end{bmatrix}+
    \begin{bmatrix}
        0\\1
    \end{bmatrix}u\,,
    \label{eq:pend_din}
\end{align}
where the state $\bx=[x_1,\, x_2]^{\top}$ contains the angle $x_1$ measured from the upright position and the angular velocity $x_2$, while the control input $u$ is proportional to a torque generated by a motor.
Our goal is to keep the inverted pendulum above the horizontal position by
satisfying ${-\pi/2 \leq x_1 \leq \pi/2}$. We consider the following constraint function from \cite{cohen2024safety,gacsi2025activated}:
\begin{equation}
    h(\bx) = \left(\frac{\pi}{2}\right)^2 - x_1^2 -\frac{1}{2\mu}\left(x_2 + Kx_1\right)^2\,,
    \label{eq:pendCBF}
\end{equation}
with parameters ${\mu>0}$, ${K>0}$, resulting in a rotated ellipse around the origin as constraint set.
Note that $h$ is a valid CBF without input bounds (for an appropriate ${\alpha \in \Ke}$)~\cite{cohen2024safety,gacsi2025activated}.
The difficulty in realizing a safety-critical controller is, however, the saturation of the input at ${u_{\min} \leq u \leq u_{\max}}$.
We address this challenge for ${u_{\min}=-0.75}$ and ${u_{\max}=1.25}$.
We design a safe controller by modifying the desired controller $k_{\mathrm{d}}(\bx)=0$ that would let the pendulum fall.

We construct a backup set-backup controller pair using the proposed method.
Note that the inverted pendulum model~\eqref{eq:pend_din} does not satisfy Assumption~\ref{assum:feedback_lin} as it has more states than inputs.
Thus, we shall select an output coordinate with a valid relative degree for feedback linearization, according to Section~\ref{sec:output_formulation}.
While choosing ${y(\bx)=x_2}$ would yield relative degree one, the resulting backup set would depend only on $x_2$ and could not to be a subset of the constraint set given by~\eqref{eq:pendCBF}.
Instead, the right choice is ${y(\bx)=x_1}$, which has relative degree two with
${L_\bf y(\bx)=x_2}$,
${L_\bg y(\bx)=0}$,
${L_\bf^2 y(\bx)=\sin(x_1)}$, and
${L_\bg L_\bf y(\bx)=1}$.
This choice satisfies Assumption~\ref{assum:output_reldeg}.

Then, with ${\boeta = [x_1-x_1^*,\, x_2]^{\top}}$, we construct the backup set in~\eqref{eq:hb_finalform}.
This is associted with the backup controller:
\begin{equation}
        k_\mathrm{b}(\bx) = \mathrm{sat} \big( -\sin(x_1) - K_1 (x_1-x_1^*) - K_2 x_2 \big) \,,
        \label{eq:pend_kb}
\end{equation}
and the corresponding linear dynamics:
\begin{equation} \label{eq:pend_linear_dyn}
    \begin{bmatrix}
        \dot{x}_1\\
        \dot{x}_2
    \end{bmatrix}=\underbrace{\begin{bmatrix}
        0 & 1\\
        -K_1 & -K_2
    \end{bmatrix}}_{\bA} \begin{bmatrix}
        x_1 - x_1^*\\x_2
    \end{bmatrix}\,.
\end{equation}

Next, we choose the parameters $c$, $\bQ$, $\bx^*$, $K_1$, and $K_2$.
We look for the equilibrium point $\bx^*$ by setting the right-hand side of \eqref{eq:pend_din} equal to zero. Solving the algebraic equations we get:
\begin{equation}
    u^* = -\sin(x_1^*)\,, \quad
    x_2^* = 0\,.
    \label{eq:EP_pend}
\end{equation}
The equilibrium point $\bx^*$ must lie inside the safe set, while the corresponding input must not saturate:  
\begin{equation}
\begin{aligned}
    \left(\frac{\pi}{2}\right)^2 - \big(x_1^*\big)^2 -\frac{1}{2\mu}\left(x_2^* + Kx_1^*\right)^2> 0\,, \\
    u_{\min} < u^* = - \sin(x_1^*) < u_{\max}.
    \label{eq:pend_eq}    
\end{aligned}
\end{equation}
%
This can be satisfied, for example, by ${x_1^*=0}$, ${x_2^*=0}$, ${u^*=0}$.
Then we design the gains $K_1$ and $K_2$ so that matrix $\bA$ in~\eqref{eq:pend_linear_dyn} is Hurwitz.
This holds for any ${K_1,K_2>0}$.
The matrix $\bP$ can be constructed by solving the CTLE in \eqref{eq:ctle} with ${\bQ=\bI}$:
\begin{equation}
    \bP=\begin{bmatrix}
        \frac{K_1(K_1+1)+K_2^2}{2K_1K_2} & \frac{1}{2K_1} \vspace{1mm} \\ 
        \frac{1}{2K_1} & \frac{K_1+1}{2K_1K_2}
    \end{bmatrix}\,.
    \label{eq:pend_P}
\end{equation}
Finally, parameter $c$ is tuned so that the backup controller does not sature over the backup set that is inside the constraint set.

Figure \ref{fig:ex2_sets} shows the constraint set $\s$, the backup set $\Sb$, and the no-saturation region of the backup controller $\Sns$ for ${K=0.15}$, ${\mu=(1-K^2)/2}$, and various choices of $K_1$, $K_2$, and $c$.
The choice of the gains $K_1$ and $K_2$ affects the shape of the sets $\Sb$ and $\Sns$.
For each gain, parameter $c$ is adjusted to ensure that the backup set lies inside the other sets.
This way, the backup set-backup controller pair is valid in each case.

Figure~\ref{fig:ex2_sets} also presents numerical simulations of the inverted pendulum~\eqref{eq:pend_din} with the backup controller~\eqref{eq:pend_kb}.
Panel (a) depicts the evolution of the backup flow for various initial conditions.
The trajectories starting from the orange set safely reach the backup set without leaving the constraint set.
Thus, the orange set represents the largest possible invariant set $\SI$ defined in~\eqref{eq:SI} with infinite integration time (${T \to \infty}$).
Panels (b)-(d) also show level sets corresponding to various finite integration times (see colorful areas).
Starting from these sets, the trajectories reach the backup set in finite time $T$, which is quantified by the numbers in seconds.
The dependency of the invariant sets on the backup controller gains is apparent.

\begin{figure}
\centering
\includegraphics[width=0.4\textwidth]{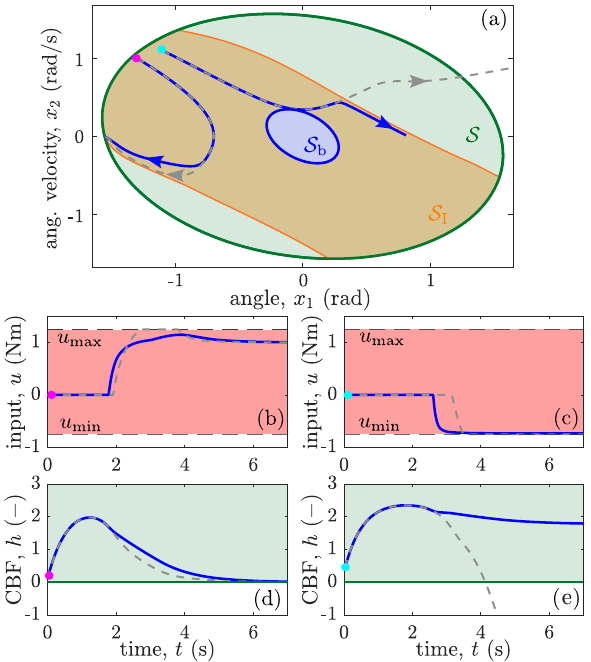}
\vspace{-2mm}
\caption{
Trajectories of the inverted pendulum~\eqref{eq:pend_din} using the backup CBF-QP~\eqref{eq:QP2} (solid blue) and the saturated counterpart of the CBF-QP~\eqref{eq:QP} (dashed gray)
for two different initial conditions (a), with the corresponding control inputs (b)-(c) and constraint functions (d)-(e).
The backup CBF-QP ensures safety with respect to the set $\SI$ (orange) that is within the constraint set $\s$ (green) while satisfying the input constraints.
The CBF-QP fails to ensure safety because the inputs are saturated to keep them within prescribed bounds.
}
\vspace{-4mm}
\label{fig:pendulum_res}
\end{figure}

Figure~\ref{fig:pendulum_res} demonstrates numerical simulation results for the inverted pendulum~\eqref{eq:pend_din} using the backup CBF-QP in~\eqref{eq:QP2} with parameters $T=5\,$s, $N_\mathrm{c}=51$, ${\alpha(h) = h}$, and ${\alpha_{\mathrm{b}}(h_{\rm b}) = h_{\rm b}}$.
The trajectories are computed for two different initial conditions which are marked with magenta and cyan dots.
The backup CBF-QP controller is compared with the traditional CBF-QP design in~\eqref{eq:QP} where, to ensure feasibility, the input constraint is omitted from the optimization and instead the input is saturated after solving the QP.
Figure~\ref{fig:pendulum_res}(a) shows the trajectories, while panels (b)-(c) present the inputs, and panels (d)-(e) depict the CBF values.
The backup CBF-QP keeps the trajectory inside the invariant set while satisfying the input constraints (see solid blue lines).
As opposed, the CBF-QP may violate safety due to the input saturation (dashed gray lines).
Panel (c) highlights the key difference between the two approaches: the backup CBF-QP responds earlier than the standard CBF approach thanks to the forward prediction of the backup flow.
Ultimately, this predictive, early response helps to maintain safety even in the presence of input constraints.
\end{example}



\section{Vehicle Braking on Asymmetric Surfaces}
\label{sec:splitmu}

We now proceed to apply the proposed method to safe vehicle braking on asymmetric (split-$\mu$) surfaces.
We focus on the realization of a safe vehicle controller which can prevent the vehicle from spinning out while the braking forces -- which are viewed as control inputs -- are kept feasible.

\subsection{Modeling}

Control design and numerical simulation require a suitable model.
This involves
vehicle, tire, and driver modeling.

\begin{figure}[t!]
\centerline{\includegraphics[width=0.5\textwidth]{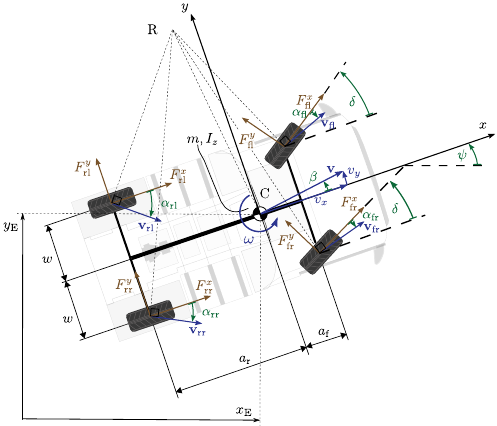}}
\vspace{-2mm}
\caption{Free-body diagram of the four-wheel planar vehicle model utilized for control design and numerical simulations.}
\vspace{-4mm}
\label{fig:FBD_vehicle}
\end{figure}

\subsubsection{Vehicle model}

To capture a braking on a split-$\mu$ surface, we use a four-wheel vehicle model \cite{jazar2025vehicle}, including left and right sides with different braking forces, and front and rear axles for considering front axle steerability.
The free-body diagram of the selected model is shown in Fig.~\ref{fig:FBD_vehicle}.
The vehicle has mass $m$ and mass moment of inertia $I_z$ about the vertical axis at the center of mass C.
The wheels are at distance $w$ from C in the lateral direction, while the front and rear axles are at distance $a_{\rm f}$ and $a_{\rm r}$ from C, respectively.
The vehicle's position in the Earth frame is given by $x_{\rm E}$ and $y_{\rm E}$, whereas the longitudinal and lateral directions are indicated by $x$ and $y$, respectively.
The yaw angle is $\psi$ and the yaw rate is $\omega$.
The velocity of the center of mass is $\mathbf{v}$, its longitudinal and lateral components are $v_x$ and $v_y$, and the side slip angle between the velocity vector and the longitudinal axis is $\beta$.
Each wheel is associated with a velocity vector $\mathbf{v}_{ij}$, a side slip angle $\alpha_{ij}$ between the velocity and the longitudinal axis of the wheels, as well as longitudinal forces $F_{ij}^{x}$ and lateral forces $F_{ij}^{y}$ acting on the wheels.
The index ${i \in \{{\rm f},{\rm r}\}}$ denotes the front or rear axle while the index ${j \in \{{\rm l},{\rm r}\}}$ refers to the left or right wheel.
We assume parallel steering with small angles, so the front wheels are at the steering angle $\delta$
relative to the longitudinal axis.

The kinematics of the vehicle are described by:
\begin{equation}
\begin{aligned}
    \dot{x}_{\rm E} & = v_x \cos\psi - v_x \tan \beta \sin\psi \,, \\
    \dot{y}_{\rm E} & = v_x \sin\psi + v_x \tan \beta  \cos\psi \,, \\
    \dot{\psi} & = \omega \,, \\
\end{aligned}
\end{equation}
which are used to calculate the trajectory of the vehicle during simulations.
The equations of motion can be derived using Lagrange's equations of the second kind \cite[Ch.~10.2]{jazar2025vehicle}:
\begin{equation}
\begin{split}
    \underbrace{\begin{bmatrix}
    \dot{v}_x\\
    \dot{\beta}\\
    \dot{\omega}
    \end{bmatrix}}_{\dot{\bx}}
    =\underbrace{
    \begin{bmatrix}
    f_v\\
    f_\beta\\
    f_\omega
    \end{bmatrix}}_{\mathbf{f}(\bx)}
    +\underbrace{\begin{bmatrix}
    \frac{\cos\delta}{m} \!&\! \frac{\cos\delta}{m} \!&\!\frac{1}{m} \!&\! \frac{1}{m}\\
    g_{1} \!&\!  g_{1} \!&\! g_{2} \!&\! g_{2}\\
    g_{3} \!&\! g_{4} \!&\! -\frac{w}{I_z} \!&\! \frac{w}{I_z}
    \end{bmatrix}}_{\mathbf{g}(\bx)}\underbrace{\begin{bmatrix}
        F_{\rm fl}^{x}\\
        F_{\rm fr}^{x}\\
        F_{\rm rl}^{x}\\
        F_{\rm rr}^{x}
    \end{bmatrix}}_{\bu}\,,
\end{split}
\label{eq:veh_sys}
\end{equation}
which are used for control design considering ${v_x>0}$ with:
\begin{equation}
\begin{aligned}
    f_v &= \omega v_x \tan\beta - \frac{\sin\delta}{m}\big(F_{\rm fl}^{y}+F_{\rm fr}^{y} \big) \,, \\
    f_\beta &= -\omega \!+\! \frac{\cos\beta}{m v_x} \big(
    (F_{\rm fl}^{y} \!+\! F_{\rm fr}^{y}) \cos(\delta\!-\!\beta) \!+\! (F_{\rm rl}^{y} \!+\! F_{\rm rr}^{y})\cos\beta \big) \,, \\
    f_\omega &= \frac{1}{I_z} \!\big(\!
    (F_{\rm fl}^{y} \!-\! F_{\rm fr}^{y}) w \sin\delta \!+\! (F_{\rm fl}^{y} \!+\! F_{\rm fr}^{y}) a_{\rm f} \cos\delta \!-\! (F_{\rm rl}^{y} \!+\! F_{\rm rr}^{y}) a_{\rm r} \big) , \\
    g_{1} & = \frac{\cos\beta}{m v_x}\sin(\delta-\beta) \,, \qquad\quad
    g_{2} = -\frac{\cos\beta}{mv_x}\sin\beta \,, \\
    g_{3} & = \frac{1}{I_z} (a_{\rm f} \sin\delta - w \cos\delta) \,, \quad
    g_{4} = \frac{1}{I_z} (a_{\rm f} \sin\delta + w \cos\delta) \,.
\end{aligned}
\end{equation}

\subsubsection{Tire model}

The tire model determines the forces acting on the wheels.
We consider the wheels in plane with forces at the wheel's contact point.
The longitudinal force takes the role of the control input.
For the lateral force, there exists a wide range of models in the literature, ranging from simple expressions to sophisticated models like Pacejka's magic formula \cite{pacejka2005tire}.
To keep the system control affine and to simplify the control design, we assume that the lateral force is not affected by the longitudinal force, and we use a linear tire model:
\begin{equation}
    F_{ij}^{y}=-C_{i} \alpha_{ij}\,,
\end{equation}
${i \in \{{\rm f},{\rm r}\}}$,  ${j \in \{{\rm l},{\rm r}\}}$, where
the wheel's cornering stiffness $C_{i}$ is considered to be the same for the left and right wheels.
The side slip angles $\alpha_{ij}$ are calculated based on kinematics \cite{jazar2025vehicle}:
\begin{equation}
    \alpha_{{\rm f}j} =
    \tan^{-1} \!\!\left(\! \frac{v_y \!+\! a_{\rm f}\omega}{v_x \!\mp\! w\omega} \!\right) \!-\! \delta\,, \;\;
    \alpha_{{\rm r}j} = \tan^{-1} \!\!\left(\! \frac{v_y \!-\! a_{\rm r}\omega}{v_x \!\mp\! w\omega} \!\right),
\end{equation}
with negative sign for left, positive for right, and ${v_y \!=\! v_x \tan\beta}$.

\subsubsection{Driver model}

The driver model assigns the steering angle $\delta$ based on the vehicle's position and orientation.
Note that the driver model is required only for numerical simulations and it is not used by the control design.
Instead, the braking controller may use steering angle $\delta$ measurements directly.
As the emphasis is not on the driver but on the controller, we choose a simple driver model on a straight road that
responds to the lateral position $y_{\rm E}$ and yaw angle $\psi$ \cite{voros2019lane}:
\begin{equation}
    \delta = -K_{y} y_{\rm E} - K_{\psi}\psi \,,
    \label{eq:ex3_driver}
\end{equation}
with driver parameters ${K_{y},K_{\psi}>0}$ that affect lateral stability.

\subsection{Safety-critical control design}

Next, we use the proposed method to design a safety-critical controller that determines the longitudinal forces $F_{ij}^{x}$ for the system~\eqref{eq:veh_sys}.
The controller seeks to prevent the vehicle from spinning out by keeping the side slip angle $\beta$ and yaw rate $\omega$ within safe bounds.
Thus, we define the constraint function:
\begin{equation}
    h(\bx)=1-\left(\frac{\beta}{\beta_{\mathrm{cr}}}\right)^2 -\left(\frac{\omega}{\omega_{\mathrm{cr}}}\right)^2\,,
    \label{eq:ex3_safeset}
\end{equation}
where $\beta_{\mathrm{cr}}$, $\omega_{\mathrm{cr}}$ are critical side slip angle and yaw rate values,
leading to an ellipse-shaped safe region in the $(\beta,\omega)$ plane.
We assume that the wheels cannot be accelerated during braking, so the longitudinal forces $F_{ij}^{x}$ cannot be positive, while there is limited friction, hence the braking forces cannot exceed a maximum value $\bar{F}_{ij}$.
This leads to the input constraints:
\begin{equation}
    -\bar{F}_{ij}\leq F_{ij}^{x}\leq 0\,.
    \label{eq:ex3_input_bounds}
\end{equation}
Importantly, the maximum braking force $\bar{F}_{ij}$ is different for the left and right wheels on asymmetric surfaces.
To minimize the stopping distance, the desired controller gives the maximum braking force,
${\kd(\bx)=[-\bar{F}_{\rm fl},\, -\bar{F}_{\rm fr},\, -\bar{F}_{\rm rl},\, -\bar{F}_{\rm rr}]^{\top}}$.


To ensure safety with input bounds, we construct a backup set-backup controller pair.
In this process, we view the steering angle $\delta$ as a parameter.
First, we choose output coordinates:
\begin{equation}
    \by(\bx)=\begin{bmatrix}
        v_x \\
        \omega
    \end{bmatrix}\,.
    \label{eq:ex3_y}
\end{equation}
Notice that we excluded the side slip angle $\beta$ from the output, because if ${\delta=\beta=0}$, the second row of $\bg(\bx)$ in \eqref{eq:veh_sys} is zero, thus ${L_\bg \beta(\bx) = \bzero}$, and $\beta$ does not have a valid relative degree.
As opposed, the output $\by$ has relative degree one.

Then, we establish a feedback linearization controller.
Based on Assumption~\ref{assum:output_reldeg}, the number of outputs and inputs must be the same.
Thus, we design only two inputs instead of four
and we assign the remaining two inputs proportionally:
\begin{equation}
    \begin{bmatrix}
        F_{\rm fl}^{x} \\
        F_{\rm fr}^{x}
    \end{bmatrix} = \hkb(\bx) \,, \quad
    \begin{bmatrix}
        F_{\rm rl}^{x} \\
        F_{\rm rr}^{x}
    \end{bmatrix} =
    \begin{bmatrix}
        r_{\rm l} F_{\rm fl}^{x}\\
        r_{\rm r} F_{\rm fr}^{x}
    \end{bmatrix}
    =\begin{bmatrix}
        r_{\rm l} \!&\! 0\\
        0 \!&\! r_{\rm r}
    \end{bmatrix}\hkb(\bx)\,,
\end{equation}
where ${r_j=\bar{F}_{{\rm r}j}/\bar{F}_{{\rm f}j}}$, ${j \in \{{\rm l},{\rm r}\}}$, are the ratios of the maximal rear and front braking forces on the left and right. This gives:
\begin{equation}
\begin{aligned}
    \dot{\bx}=\bf(\bx) + 
    \bg(\bx) \kb(\bx)
    =\bf(\bx) + \hat{\bg}(\bx) \hkb(\bx), \\
    \kb(\bx) =
    \begin{bmatrix}
    \bI \\ \mathbf{r}
    \end{bmatrix}
    \hkb(\bx) \,, \quad
    \hat{\bg}(\bx) =
    \bg(\bx)
    \begin{bmatrix}
    \bI \\ \mathbf{r}
    \end{bmatrix} \,,
    \label{eq:ex3_mod_dyn}
\end{aligned}
\end{equation}
where $\bI$ is the ${2 \times 2}$ identity matrix and ${\mathbf{r} = {\rm diag}\{r_{\rm l},r_{\rm r}\}}$.
As the output in \eqref{eq:ex3_y} has relative degree one, the backup controller can be computed using \eqref{eq:feedback_lin}:
\begin{equation}
    \hkb(\bx)=\mathrm{sat}\bigg(\!
    \begin{bmatrix}
    \frac{\cos\delta + r_{\rm l}}{m} \!&\! \frac{\cos\delta + r_{\rm r}}{m} \\
    g_{3} \!-\! \frac{r_{\rm l} w}{I_z} \!&\! g_{4} \!+\! \frac{r_{\rm r} w}{I_z}
    \end{bmatrix}^{-1}\!\!
    \begin{bmatrix}
    -f_v - a_x^* \\ 
    -f_\omega \!-\! K_\omega (\omega \!-\! \omega^*)
    \end{bmatrix} \!\bigg),
    \label{eq:ex3_kb_virt}
\end{equation}
where
$a_x^*$ is a desired deceleration to be selected later.
This yields the closed-loop output dynamics with the gain ${K_\omega>0}$:
\begin{equation} \label{eq:ex3_closedloop}
     \begin{bmatrix}
         \dot{v}_x\\
         \dot{\omega}
     \end{bmatrix}=
     \begin{bmatrix}
         -a_x^* \\
         -K_\omega (\omega-\omega^*)
     \end{bmatrix}\,.
\end{equation}


Next, we construct the backup set based on~\eqref{eq:hb_finalform}:
\begin{equation}
    \hb(\bx) = c- p_{\beta}\big(\beta-\beta^*\big)^2 - p_{\omega}(\omega-\omega^*)^2 \,.
\label{eq:ex3_mod_hb}
\end{equation}
This formula is a modified version of~\eqref{eq:hb_finalform}, as we incorporate the state $\beta$ instead of the output $v_x$ when constructing the backup set $\Sb$.
While this modification allows us to ensure that ${\Sb\subseteq\s}$, it may affect the forward invariance of $\Sb$ under $\kb$.
Therefore, forward invariance will be verified later for the appropriate choice of ${c>0}$, $\beta^*$, $\omega^*$, and ${p_{\beta},p_\omega>0}$.

\begin{figure}
\centering
\includegraphics[width=0.99\linewidth]{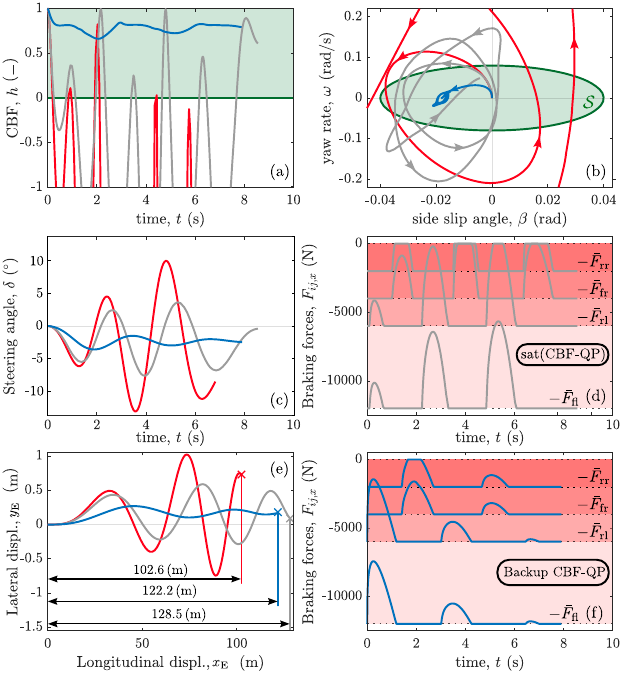}
\vspace{-7mm}
\caption{Simulation of the vehicle~\eqref{eq:veh_sys} braking on a split-$\mu$ surface using maximum braking forces (red), the saturated counterpart of the CBF-QP~\eqref{eq:QP} (gray), and the proposed backup CBF-QP~\eqref{eq:QP2} (blue). These methods are compared based on constraint function (a), trajectory in the $(\beta,\omega)$-plane (b), steering angle (c), lateral and longitudinal displacement (e), and braking forces (d)-(f). While the maximum braking forces minimize the stopping distance, they excessively violate safety.
Likewise, the CBF-QP is unsafe due to the saturation, and it increases the braking distance.
The proposed backup CBF-QP ensures safe input-bounded braking with a trade-off in stopping distance.} 
\vspace{-4mm}
\label{fig:splitmu_sim}
\end{figure}

Next, we select the parameters $c$, $p_\beta$, $p_\omega$, $\beta^*$, $\omega^*$, $a_x^*$, and $K_\omega$.
To select the equilibrium point $\beta^*$ and $\omega^*$, we observe that constant-speed rectilinear motion is described by ${\beta=0}$, ${\omega=0}$, ${\delta = 0}$, and zero braking force.
However, this cannot be chosen directly as equilibrium, because the driver may set nonzero $\delta$  and also the zero force violates~\eqref{eq:equilibrium_no_saturation} as it is equal to the maximum input.
Instead, we choose ${\omega^*=0}$ at the center of the constraint set and we seek for a nonzero $\beta^*$ that satisfies ${\dot{\beta}=0}$ in~\eqref{eq:veh_sys}.
For ${\omega=0}$, we have ${\alpha_{{\rm f}j} = \beta-\delta}$, ${\alpha_{{\rm r}j} = \beta}$, and with small angle approximations for $\beta$ and $\delta$ we obtain:
\begin{equation} \label{eq:small_angle_approx}
\begin{aligned}
    f_v &\approx 0 \,, \;\;
    f_\beta \approx \frac{-2}{m v_x} \big( C_{\rm f}(\beta-\delta) + C_{\rm r}\beta \big) \,, \;\;
    g_{1} \approx 0 \,, \;\;
    g_{2} \approx 0 \,, \\
    f_\omega &\approx -\frac{2}{I_z} \big( C_{\rm f} (\beta-\delta) a_{\rm f} - C_{\rm r} \beta a_{\rm r} \big) , \;\;
    g_{3} \approx -\frac{w}{I_z} \,, \;\;
    g_{4} \approx \frac{w}{I_z} \,.
\end{aligned}
\end{equation}
With this approximation, ${\dot{\beta}=0}$ yields the choice:
\begin{equation}
   \beta^*=\frac{C_{\rm f}}{C_{\rm f}+C_{\rm r}}\delta\,.
   \label{eq:ex3_beta_star}
\end{equation}
As $\beta^*$ depends on $\delta$, the backup set in \eqref{eq:ex3_mod_hb} moves in the state space if $\delta$ changes over time.
Since the center of the backup set must be inside the constraint set according to~\eqref{eq:equilibrium_safety}, this choice of $\beta^*$ is valid if ${-\beta_{\mathrm{cr}}<\beta^*<\beta_{\mathrm{cr}}}$.
Thus, our method is applicable in a narrower range of steering angles.

Now we determine the parameter $a_x^*$ that affects the evolution of the velocity $v_x$ in~\eqref{eq:ex3_closedloop}.
We notice that, while the backup set $\Sb$ must be forward invariant, the choice of $a_x^*$ does not need to render $v_x$ exponentially stable, because the definition~\eqref{eq:ex3_mod_hb} of $\Sb$ does not involve $v_x$.
Instead, we can just ensure ${\dot{v}_x \leq 0}$ for braking by choosing ${a_x^*\geq0}$.
We must select $a_x^*$ so that the braking forces in~\eqref{eq:ex3_kb_virt} do not saturate inside the backup set $\Sb$.
We choose $a_x^*$ such that the saturation curves at zero force move together with $\Sb$ as $\delta$ changes, by setting:
\begin{equation} \label{eq:saturation_betad}
\begin{aligned}
    k_{\mathrm{FL},1}(v_x,\beta^*+\beta_\mathrm{d},0)=0\,, \quad \mathrm{if}\, \delta\geq 0 \,, \\
    k_{\mathrm{FL},2}(v_x,\beta^*-\beta_\mathrm{d},0)=0\,, \quad \mathrm{if}\, \delta < 0 \,,
\end{aligned}
\end{equation}
where ${\beta_\mathrm{d}>0}$ is a parameter that represents a fixed distance between the center $\beta^*$ of the backup set and the saturation curves.
We solve~\eqref{eq:saturation_betad} with the approximations in~\eqref{eq:small_angle_approx} and ${\cos\delta \approx 1}$, and we derive the formula:
\begin{equation}
    a_x^*=\frac{2}{mw}\bigg(
    \frac{a_{\rm f}+a_{\rm r}}{\frac{1}{C_{\rm f}} + \frac{1}{C_{\rm r}}} |\delta|
    + (C_{\rm r}a_{\rm r}-C_{\rm f}a_{\rm f}) \beta_{\rm d}
    \bigg)\,,
    \label{eq:ex3_a_x}
\end{equation}
which is positive for any ${\beta_{\rm d}>0}$ if
${C_{\rm r} a_{\rm r} - C_{\rm f} a_{\rm f}>0}$.
This $a_x^*$ prevents the saturation of braking forces at the equilibrium point and ensures the satisfaction of~\eqref{eq:equilibrium_no_saturation} for any $\delta$ of interest.

\begingroup
\setlength{\tabcolsep}{4pt}
\begin{table}
\caption{Parameters of the split-$\mu$ braking simulation}
\vspace{-5mm}
\begin{center}
\begin{tabular}{|c | c | c| |c | c | c| |c | c | c|} 
\hline
Par. & Value & Unit &
Par. & Value & Unit &
Par. & Value & Unit \\
\hline
$m$ & $8850$ & kg &
$K_y$ & 0.2 & rad/m &
$\beta_{\mathrm{d}}$ & 0.016 & rad \\ 
$I_z$ & $36950$ & kgm$^2$ &
$K_{\psi}$ & 0.4 & 1 &
$p_{\beta}$ & $1$ & 1 \\
$w$ & $1.5$ & m &
$\beta_{\mathrm{cr}}$ & $0.04$ & rad & 
$K_\omega$ & $1$ & 1/s \\
$a_{\rm f}$ & $1.4$ & m &
$\omega_{\mathrm{cr}}$ & $0.08$ & rad/s & 
$c$ & $5\!\cdot\!10^{-5}$ & 1 \\
$a_{\rm r}$ & $1.6$ & m &
$\bar{F}_{\rm fl}$ & $12$ & kN &
$T$ & $0.1$ & s \\ 
$C_{\rm f}$ & $130$ & kN/rad &
$\bar{F}_{\rm fr}$ & $4$ & kN &
$N_{\mathrm{c}}$ & $200$ & 1 \\ 
$C_{\rm r}$ & $175$ & kN/rad &
$\bar{F}_{\rm rl}$ & $6$ & kN &
$\alpha(h)$ & $8h$ & 1/s \\
$v_{x,0}$ & $25$ &m/s &
$\bar{F}_{\rm rr}$ & $2$ & kN &
$\alpha_{\rm b}(h_{\rm b})$ & $25h_{\rm b}$ & 1/s \\
\hline
\end{tabular}
\end{center}
\vspace{-5mm}
\label{tab:parameters}
\end{table}
\endgroup

Finally, we choose
${p_\omega=1/(2 K_\omega)}$ based on the CTLE \eqref{eq:ctle} with ${\bQ=\bI}$.
Further, any ${c, p_\beta, K_\omega>0}$ may be chosen as long as the set $\Sb$ is forward invariant, ${\Sb\subseteq\s}$, and ${\Sb \subseteq \Sns}$.
These properties will be verified by plotting $\Sb$, $\s$, and $\Sns$. 

Having constructed the proposed backup controller~\eqref{eq:ex3_kb_virt} and backup set~\eqref{eq:ex3_mod_hb} offline, we use the backup CBF-QP in~\eqref{eq:QP2} to calculate a safe and feasible control input online.
This includes the forward prediction of the vehicle's motion under the backup controller.
The prediction depends on the driver's future behavior through the future values of the steering angle $\delta$.
Since these values are unknown, we calculate the prediction by assuming that the steering angle remains constant over the prediction horizon, using the current value of $\delta$.
Accordingly, the parameters $\beta^*$ in~\eqref{eq:ex3_beta_star} and $a_x^*$ in~\eqref{eq:ex3_a_x}, that depend on the steering angle $\delta$, are also considered to be constants over the horizon, and we neglect their time derivatives ($\dot{\beta}^*\approx 0$, $\dot{a}_x^*\approx 0$) in $\dot{h}$ and $\dot{h}_{\rm b}$ during the implementation of~\eqref{eq:QP2}.



\subsection{Simulation results}

\begin{figure*}
\centering
\includegraphics[width=1\linewidth]{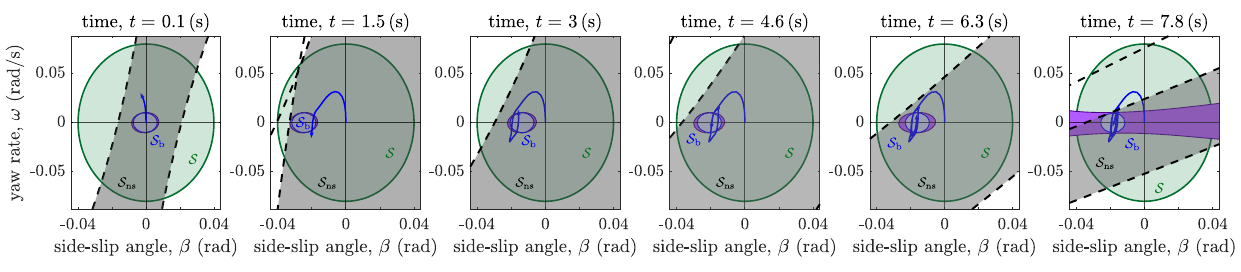}
\vspace{-7mm}
\caption{Evolution of the backup set $\Sb$ (blue) during split-$\mu$ braking relative to the simulated trajectory (blue line), constraint set $\s$ (green), no-saturation set $\Sns$ (gray), and the region where ${\dot{h}_{\mathrm{b}} \geq-\alpha_{\mathrm{b}} (h_{\mathrm{b}})}$ (purple).
The steering angle changes according to the driver model~\eqref{eq:ex3_driver}, and the backup set moves along the horizontal axis.
Throughout the motion, $\Sb$ remains inside $\s$, $\Sns$ and the purple area, indicating that it is a valid, forward invariant backup set.
}
\vspace{-4mm}
\label{fig:splitmu_sets}
\end{figure*}

Figure~\ref{fig:splitmu_sim} depicts numerical simulations of the split-$\mu$ braking scenario considering a cabover truck (see Fig.~\ref{fig:musplitvehicle}) with the parameters listed in Table~\ref{tab:parameters}.
At the beginning, we assume that the truck drives onto the split-$\mu$ surface with initial conditions ${v_x(0)=v_{x,0}}$, ${\beta(0)=0}$, and ${\omega(0) = 0}$, exerting maximal braking forces given by $\kd(\bx)$.
We consider that the calculated braking forces are applied instantly at the wheels, which can be realized by ABS in practice.
The driver parameters $K_y$ and $K_\psi$ are selected so that initially at the highest speed they yield unstable behavior based on linear stability analysis~\cite{voros2019lane}, and they gradually become stable as $v_x$ decreases, imitating a realistic driver.
Therefore, the intervention of the safety-critical controller is crucial at high speeds for maintaining bounded lateral motions.
The constraint set parameters $(\beta_{\rm cr},\omega_{\rm cr})$ in~\eqref{eq:ex3_safeset} are chosen based on the results in~\cite{mirzaeinejad2014optimization}. 
We select the parameters $(\beta_{\rm d}, p_\beta, K_{\omega}, c)$ so that they yield a valid backup set-backup controller pair, which will be demonstrated in Fig.~\ref{fig:splitmu_sets}. We choose a relatively short prediction horizon $T$, since the accuracy of the forward prediction may deteriorate over long horizons under the assumption of constant steering angle $\delta$.  

Figure~\ref{fig:splitmu_sim} compares three braking strategies:
(i) the desired controller $\kd$ with maximal braking forces (red lines) to minimize the stopping distance (also called {\em select high} method);
(ii) the CBF-QP design~\eqref{eq:QP} without incorporating input constraints into the optimization and instead saturating the controller (gray lines); and
(iii) the backup CBF-QP~\eqref{eq:QP2} with the proposed backup set-backup controller pair (blue lines).
Figure~\ref{fig:splitmu_sim} plots the constraint function in panel (a), the trajectories in the ($\beta$,$\omega$) plane in panel (b), the steering angles set by the driver in panel (c), and the vehicle's position in panel (e).
Panels (d) and (f) present the braking forces, where dotted lines show the maximal values.
While the select high strategy achieves the shortest possible stopping distance, it results in excessive lateral motions.
In comparison, the saturated CBF-QP reduces the side slip angle and the yaw rate at the price of a longer stopping distance, but it still violates the safety constraints due to the saturation.
The proposed backup CBF-QP not only keeps the system safe and simultaneously satisfies the input bounds, but also yields the smallest steering angles and lateral displacement indicating a smoother and more stable vehicle response.
The corresponding stopping distance lies between those of the other two controllers, achieving a trade-off between feasible safety and short stopping distance.
Notice that, similar to Example~\ref{ex:pendulum}, the backup CBF-QP intervenes earlier than the saturated CBF-QP thanks to forward prediction.

Figure~\ref{fig:splitmu_sets} shows the backup set in~\eqref{eq:ex3_mod_hb} at different time instants of the simulation.
While the backup set $\Sb$ (blue) moves along the horizontal axis as the steering angle $\delta$ changes, it remains within both the constraint set $\s$ (green) and the no-saturation region $\Sns$ (gray).
Note that, since there are two independent backup control inputs (given by $\hat{k}_{\mathrm{b},1}$ and $\hat{k}_{\mathrm{b},2}$), the set $\Sns$ has four saturation bounds, two of which follow the change of $\delta$ by keeping a fixed distance $\beta_{\rm d}$ from the center $\beta^*$ of the backup set according to~\eqref{eq:saturation_betad}.
Furthermore, we visualize the set of points $\bx$ where ${\dot{h}_{\rm b}(\bx,\mathbf{k}(x)) \geq-\alpha_{\rm b}\big(h_{\rm b}(\bx)\big)}$ holds (purple).
Because this set encapsulates the backup set at all times, we conclude that the backup set is indeed forward invariant and it is valid as desired.


\label{sec:sim}

\section{Conclusions}
\label{sec:concl}

In this paper, we developed a safety-critical control framework for road vehicles to ensure bounded lateral motions with minimal stopping distance during split-$\mu$ braking.
The fact that braking forces are limited by friction, and that vehicles cannot be accelerated during heavy braking, led to input constraints in addition to lateral safety constraints.
We employed the backup set method to synthesize controllers that guarantee the satisfaction of these constraints.
We introduced a novel, systematic approach to generating the required valid backup set-backup controller pairs, using a combination of feedback linearization and continuous time Lyapunov equations, along with meticulously chosen parameters.
We demonstrated our method on simple examples and highlighted its advantages over traditional CBF formulations.
Finally, we adapted the method to split-$\mu$ braking and showcased that the proposed controller achieves safety with limited braking forces while maintaining a short stopping distance. 

\bibliographystyle{IEEEtran}
\bibliography{limited_braking}	

\newpage
\vspace{-10mm}
\begin{IEEEbiography}[{\includegraphics[width=1in,height=1.25in,clip,keepaspectratio]{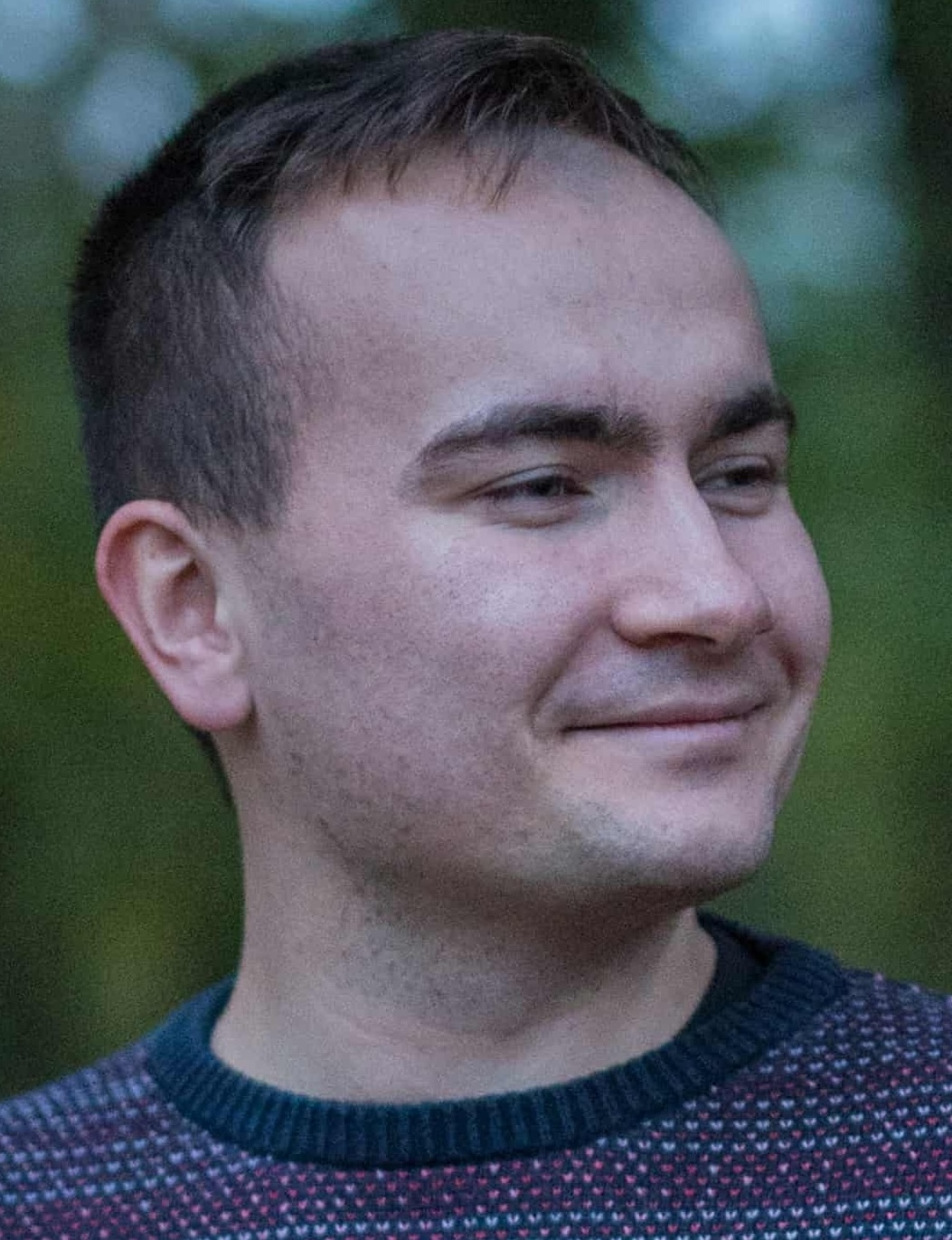}}]
 {Laszlo Gacsi} received the B.Sc. and M.Sc. degrees
 in mechanical engineering from the Budapest University of Technology and Economics,
 Budapest, Hungary, in 2022 and 2024, and he is
 currently pursuing the Ph.D. degree at the Wichita State University. His research interests include nonlinear dynamics and control, safety-critical control with applications to autonomous vehicles.
 \end{IEEEbiography}

 \vspace{-120mm}
 \begin{IEEEbiography}[{\includegraphics[width=1in,height=1.25in,clip,keepaspectratio]{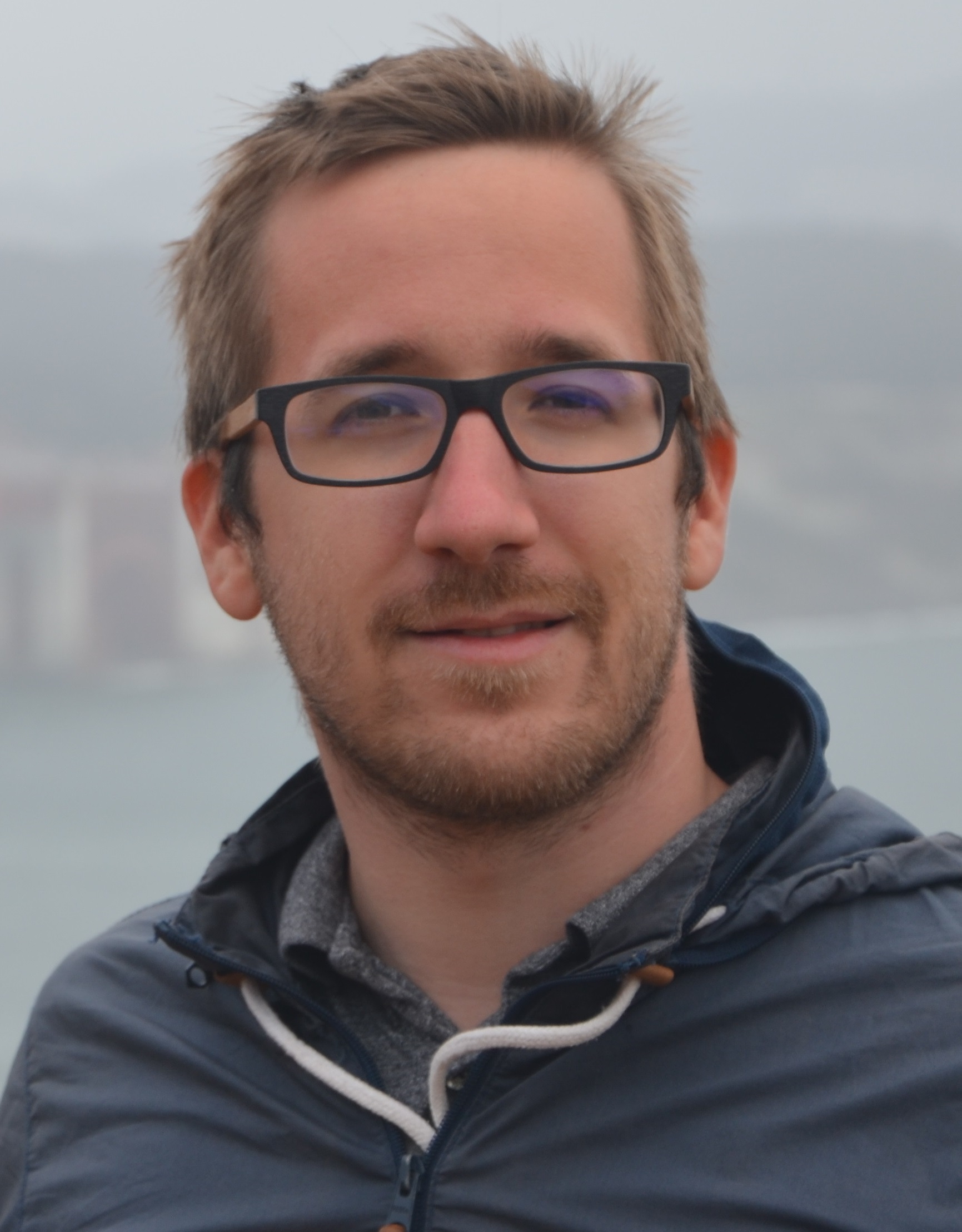}}]
 {Adam K. Kiss} received his B.Sc., M.Sc. and Ph.D. degrees in mechanical engineering from the Budapest University of Technology and Economics (BME) in 2013, 2015 and 2023.
 He is currently a researcher at the HUN-REN–BME Dynamics of Machines Research Group. He has been involved in EU-funded projects on machine tool vibrations. His current research interests include nonlinear dynamics, safety-critical control and time delay systems with applications to machine tool vibrations and connected automated vehicles.
 \end{IEEEbiography}

\vspace{-120mm}
 \begin{IEEEbiography}[{\includegraphics[width=1in,height=1.25in,clip,keepaspectratio]{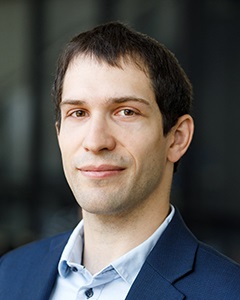}}]
 {Tamas G. Molnar} is an Assistant Professor of Mechanical Engineering at the Wichita State University since 2023. Beforehand, he held postdoctoral positions at the California Institute of Technology, from 2020 to 2023, and at the University of Michigan, Ann Arbor, from 2018 to 2020. He received his PhD and MSc in Mechanical Engineering and his BSc in Mechatronics Engineering from the Budapest University of Technology and Economics, Hungary, in 2018, 2015, and 2013. His research interests include nonlinear dynamics and control, safety-critical control, and time delay systems with applications to connected automated vehicles, robotic systems, and autonomous systems.
 \end{IEEEbiography}

\end{document}